\documentclass[11pt,letterpaper]{article}

\usepackage{etoolbox}

\usepackage[margin=1in]{geometry}

\usepackage[utf8]{inputenc}

\usepackage{amsmath}
\usepackage{amssymb}
\usepackage{amsthm}
\usepackage{bbm}
\usepackage{thm-restate}
\newtheorem{theorem}{Theorem}[section]
\newtheorem{lemma}[theorem]{Lemma}

\usepackage{url}
\usepackage[table]{xcolor}
\usepackage{array}

\usepackage{setspace}

\usepackage{wrapfig}

\makeatletter
\newcommand{\labeltarget}[1]{\Hy@raisedlink{\hypertarget{#1}{}}}
\makeatother

\usepackage{lmodern}
\usepackage{times}
\usepackage{upgreek}

\usepackage{complexity}

\usepackage[inline]{enumitem}
\usepackage{moreenum}

\usepackage{mathtools}
\usepackage[super]{nth}

\usepackage{bm}

\usepackage{xspace}

\usepackage{ifthen}

\usepackage[immediate]{silence}
\usepackage{sectsty}
\DeactivateWarningFilters[tmp]
\allsectionsfont{\boldmath}

\setlist[enumerate]{nosep,topsep=0.1em}
\setlist[enumerate,1]{label=(\roman*), leftmargin=2.2em}
\setlist[itemize]{nosep,topsep=0.3em}

\usepackage[textsize=footnotesize, color=blue!30!white]{todonotes}

\usepackage{xfrac}

\usepackage{tikz}
\usetikzlibrary{calc}
\usetikzlibrary{shapes.geometric}
\usetikzlibrary{arrows}
\usetikzlibrary{decorations.pathreplacing, decorations.pathmorphing}

\usepackage{tcolorbox}
\tcbuselibrary{skins}

\usepackage{float}
\usepackage{graphicx}
\graphicspath{{graphics/}}
\makeatletter
\newcommand\appendtographicspath[1]{%
  \g@addto@macro\Ginput@path{#1}%
}
\makeatother

\usepackage[margin=10pt,font=small,labelfont=bf,skip=-5pt]{caption}
\usepackage{subcaption}

\usepackage[
backend=biber,
style=alphabetic,
citestyle=alphabetic,
maxalphanames=6,
maxcitenames=99,
mincitenames=98,
maxbibnames=99,
giveninits=true,
]{biblatex}

\usepackage{xpatch}

\makeatletter
\patchcmd\blx@bblinput{\blx@blxinit}
                      {\blx@blxinit
                      }{}{\fail}
\makeatother

\usepackage[vlined,ruled,algo2e]{algorithm2e}
\SetAlgoSkip{bigskip}
\SetAlgoInsideSkip{smallskip}

\usepackage{mdframed}

\usepackage{bigfoot}
\usepackage[pdfpagelabels,bookmarks=false]{hyperref}
\hypersetup{colorlinks, linkcolor=darkblue, citecolor=darkgreen, urlcolor=darkblue}

\definecolor{darkblue}{rgb}{0,0,0.38}
\definecolor{darkred}{rgb}{0.8,0,0}
\definecolor{darkgreen}{rgb}{0.1,0.35,0}

\DeclareMathOperator{\apex}{apex}

\newcommand\OPT{\ensuremath{\mathrm{OPT}}}
\newcommand\Drop{\ensuremath{\mathrm{Drop}}}
\renewcommand{\epsilon}{\varepsilon}

\def\cupp{\stackrel{.}{\cup}}
\def\bigcupp{\stackrel{.}{\bigcup}}

\makeatletter
\let\@@pmod\pmod
\DeclareRobustCommand{\pmod}{\@ifstar\@pmods\@@pmod}
\def\@pmods#1{\mkern8mu({\operator@font mod}\mkern 6mu#1)}
\makeatother

\makeatletter
\let\@@mod\mod
\DeclareRobustCommand{\mod}{\@ifstar\@mods\@@mod}
\def\@mods#1{\mkern8mu{\operator@font mod}\mkern 6mu#1}
\makeatother

\def\Cscr{\mathcal{C}}

\def\Mscr{\mathcal{M}}

\def\Pscr{\mathcal{P}}

\newbool{soda_submission}

\makeatletter
\def\@fnsymbol#1{\ensuremath{\ifcase#1\or *\or %
\ddagger\or
    \mathsection\or \mathparagraph\or \|\or **\or \dagger\dagger
    \or \ddagger\ddagger \else\@ctrerr\fi}}
\makeatother

\title{Local Search for Weighted Tree Augmentation and Steiner Tree%
\ifbool{soda_submission}{}{\thanks{
This project received funding from Swiss National Science Foundation grant 200021\_184622 and the European Research Council (ERC) under the European Union's Horizon 2020 research and innovation programme (grant agreement No 817750).
}}%
} 

\ifbool{soda_submission}{
\author{}
}{
\author{
Vera Traub\thanks{
Department of Mathematics, ETH Zurich, Zurich, Switzerland.
Email: \href{mailto:vera.traub@ifor.math.ethz.ch}%
{vera.traub@ifor.math.ethz.ch}.
}
\and
Rico Zenklusen\thanks{
Department of Mathematics, ETH Zurich, Zurich, Switzerland.
Email: \href{mailto:ricoz@ethz.ch}%
{ricoz@ethz.ch}.}
}
}

\date{}

\begin{document}

\maketitle
\thispagestyle{empty}
\addtocounter{page}{-1}
\enlargethispage{-1cm}

\begin{abstract}
We present a technique that allows for improving on some relative greedy procedures by well-chosen (non-oblivious) local search algorithms. Relative greedy procedures are a particular type of greedy algorithm that start with a simple, though weak, solution, and iteratively replace parts of this starting solution by stronger components. Some well-known applications of relative greedy algorithms include approximation algorithms for Steiner Tree and, more recently, for connectivity augmentation problems.

The main application of our technique leads to a $(1.5+\epsilon)$-approximation for Weighted Tree Augmentation, improving on a recent relative greedy based method with approximation factor $1+\ln 2 + \epsilon\approx 1.69$.
Furthermore, we show how our local search technique can be applied to Steiner Tree, leading to an alternative way to obtain the currently best known approximation factor of $\ln 4 + \epsilon$. Contrary to prior methods, our approach is purely combinatorial without the need to solve an LP. Nevertheless, the solution value can still be bounded in terms of the well-known hypergraphic LP, leading to an alternative, and arguably simpler, technique to bound its integrality gap by $\ln 4$.
\end{abstract}
 
\ifbool{soda_submission}{}{
\begin{tikzpicture}[overlay, remember picture, shift = {(current page.south east)}]
\begin{scope}[shift={(-1.1,2.5)}]
\def\hd{2.5}
\node at (-2*\hd,0) {\includegraphics[height=0.5cm]{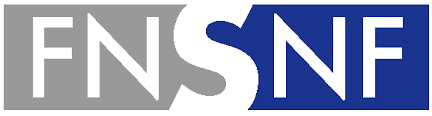}};
\node at (-\hd,0) {\includegraphics[height=1.0cm]{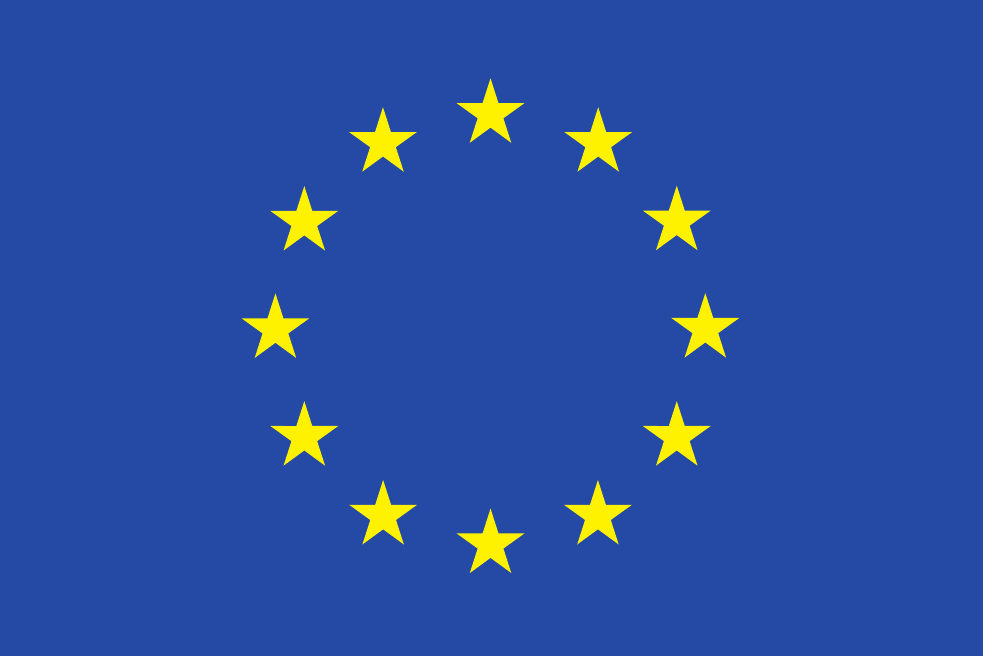}};
\node at (-0.2*\hd,0) {\includegraphics[height=1.2cm]{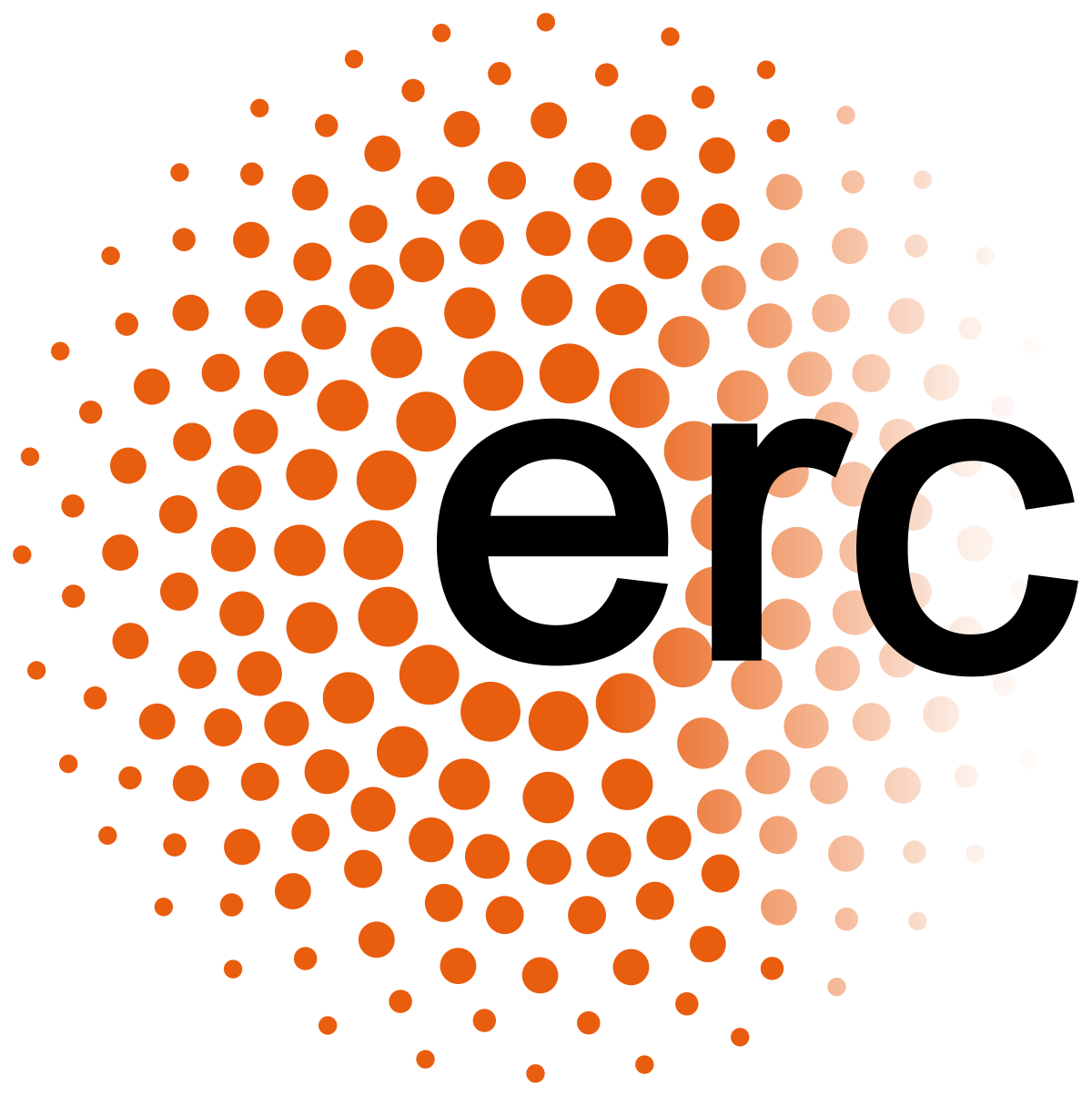}};
\end{scope}
\end{tikzpicture}
}

\clearpage

\section{Introduction}\label{sec:intro}

Relative greedy procedures, also known as relative greedy heuristics, have been employed in particular for different covering problems to obtain improved approximation guarantees compared to what one would obtain with a more canonical approach.
They follow a simple yet powerful strategy.
They start with a simple well-structured solution, with a weak approximation guarantee, and then successively improve this starting solution by replacing parts of it with cheaper components.
This approach was first used by \textcite{zelikovsky_1996_better} in the context of the Steiner tree problem.
More precisely, \citeauthor{zelikovsky_1996_better} started with a minimum spanning tree over the terminals, which is a well-known $2$-approximation for the Steiner tree problem, and then replaced some of the spanning tree edges by cheap Steiner components, i.e., subgraphs that connect several terminals at small cost.
This led to a method with an approximation guarantee of $1+\ln 2 + \epsilon < 1.7$.
(We remark that there have been later improvements in the approximation guarantee for Steiner Tree, leading to the currently best factor of $\ln4 + \epsilon < 1.39$ by \textcite{byrka_2013_steiner} (see also~\cite{goemans_2012_matroids}).)

More recently, starting with work of \textcite{cohen_2013_approximation}, relative greedy algorithms have found several applications in connectivity augmentation problems and beyond (see also~\cite{nutov_2019_approximating,nutov_2020_approximation,traub_2021_better}).
In particular, very recently, a first better-than-two approximation for Weighted Tree Augmentation was derived through a relative greedy approach by \textcite{traub_2021_better}.
The Weighted Tree Augmentation Problem (WTAP) is a very elementary and heavily studied connectivity augmentation problem defined as follows.
Given is a spanning tree $G=(V,E)$ with a set $L\subseteq \left(\!\begin{smallmatrix} V\\ 2 \end{smallmatrix}\!\right)$ of candidate edges to be added to $G$, which are also called \emph{links}, and positive link weights $w\colon L \to \mathbb{R}_{> 0}$.
The task is to find a minimum weight link set $F\subseteq L$ such that the graph $(V,E\cup F)$ is $2$-edge-connected.
WTAP can easily be seen to capture the problem of increasing the edge-connectivity of an arbitrary connected graph $G$ from $1$ to $2$, because one can contract all $2$-edge-connected components of $G$ to obtain a spanning tree.
More generally, also the problem of increasing the edge-connectivity of a $k$-edge-connected graph from $k$ to $k+1$ can be reduced to WTAP whenever $k$ is odd (see, e.g.,~\cite{cheriyan_1999_2-coverings}).

In a similar spirit as for the Steiner tree problem, the relative greedy algorithm of~\cite{traub_2021_better} for WTAP starts with a simple WTAP solution $F_0\subseteq L$ that is only guaranteed to be a $2$-approximation.
The solution $F_0$ then gets iteratively improved.
In a general iteration, the algorithm has a current WTAP solution of the form $(F_0\setminus D)\cup C$, where $D\subseteq F_0$ is a set of links from the initial solution that have been replaced in prior iterations by a cheaper set $C\subseteq L$, i.e., $w(C)<w(D)$.\footnote{When formalizing relative greedy algorithms, it is sometimes convenient to allow $w(C)=w(D)$, in which case no strict improvement is obtained. However, for this brief sketch of how relative greedy algorithms work, we do not consider this possibility.}
The algorithm then seeks to find, \emph{only among the not-yet-removed links of $F_0$}, i.e., $F_0\setminus D$, a set $\overline{D}\subseteq F_0 \setminus D$ together with a replacement set $\overline{C}\subseteq L$ such that
\begin{enumerate}
\item $(F_0\setminus (D\cup \overline{D})) \cup (C \cup \overline{C})$ is a WTAP solution, and
\item $w(\overline{C}) < w(\overline{D})$,
\end{enumerate}
which implies that the new WTAP solution is cheaper than the previous one, i.e., $w((F_0\setminus (D\cup \overline{D})) \cup (C \cup \overline{C})) < w((F_0\setminus D) \cup C)$.
Hence, the crucial challenge in designing relative greedy procedures is to show the existence of such an improving pair $(\overline{C}, \overline{D})$ (and obtain a way to efficiently find such a pair).
This is typically done through an averaging argument.
For example, in the context of WTAP, an optimal WTAP solution $\OPT\subseteq L$ is carefully partitioned into well-structured components $C_1,\ldots, C_p$. 
Then the existence of a good improving pair is implied by showing that a randomly selected component $\overline{C}$ among $C_1,\ldots, C_p$ allows for removing a link set $\overline{D}\subseteq F_0\setminus C$ such that, in expectation, $w(\overline{C}) < w(\overline{D})$.

An arguably weak spot of relative greedy approaches is that they only seek improvements with respect to the part of the solution that is left from the starting solution, i.e., $F_0\setminus D$.
Hence, potential gains that could be obtained by possibly removing some of the links in $C$, which got added later, are not considered.
This is because links in $C$ are much less structured than the once carefully chosen in the starting solution, making it difficult to develop methods that gain on those links.

The goal of this paper is to show how local search algorithms can be designed that address this problem, leading to stronger approximation guarantees and further insights.
The local search approach we suggest is a so-called non-oblivious one.
More precisely, instead of measuring progress of the approach solely in terms of how the value of the current solution improves, we introduce a well-chosen potential function that, loosely speaking, also measures whether a replacement step leads to a solution that is easier to improve in future iterations.

\subsection{Our results}

The main result based on our new non-oblivious local search approach is the currently best approximation algorithm for WTAP.
\begin{theorem}\label{thm:mainWTAP}
For any $\epsilon >0$, there is a $(1.5+\epsilon)$-approximation algorithm for WTAP.
\end{theorem}
This result improves on a recent approximation algorithm for WTAP with approximation guarantee $1+\ln 2 + \epsilon\approx 1.69$~\cite{traub_2021_better}.
Until recently, no better-than-two approximation was known for WTAP, unlike for its unweighted version where all links have unit weight, which is often simply called the Tree Augmentation Problem (TAP).
For TAP, several approaches have been developed that reach approximation factors of $1.5$ or $1.5+\epsilon$, respectively~\cite{kortsarz_2016_simplified,cheriyan_2018_approximating_b,fiorini_2018_approximating}.%
\footnote{The $(1.5+\epsilon)$-approximation by~\textcite{fiorini_2018_approximating}, which builds up on prior work by~\textcite{adjiashvili_2018_beating}, even works for WTAP as long as the ratio between largest to smallest weight is bounded by a constant.
Later results by~\textcite{grandoni_2018_improved} and~\textcite{cecchetto_2021_bridging} allow for obtaining factors below $1.5+\epsilon$ for this case.
Moreover, \textcite{nutov_2017_tree} presented a technique with which the $(1.5+\epsilon)$-approximation of~\cite{fiorini_2018_approximating} can be extended to instances where the ratio of largest to smallest weight is logarithmic in the number of vertices.}
Only recently, approximation guarantees below $1.5$~have been achieved for TAP~\cite{grandoni_2018_improved,cecchetto_2021_bridging}, with the currently best factor being~$1.393$~\cite{cecchetto_2021_bridging}. Theorem~\ref{thm:mainWTAP} narrows the gap between the unweighted and weighted version significantly.
We highlight that the canonical LP relaxation for WTAP, known as the \emph{cut LP}, is known to have an integrality gap of at least $1.5$~\cite{cheriyan_2008_integrality}.
Hence, if, contrary to our combinatorial approach, one would like to design an LP-based approach for WTAP improving on Theorem~\ref{thm:mainWTAP}---beyond removing the arbitrarily small error $\epsilon$---a stronger LP relaxation would be needed.

Leveraging the same non-oblivious local search technique as we use for WTAP, we present a local search algorithm for Steiner Tree, which, for every constant $\epsilon >0$, leads to an $(\ln 4 + \epsilon)$-approximation.
This matches the currently best approximation for the Steiner tree Problem~\cite{byrka_2013_steiner} (see also~\cite{goemans_2012_matroids} for a variation of the approach in~\cite{byrka_2013_steiner} with an LP-based analysis).
Even though we reuse key insights of prior approaches, our local search procedure, contrary to prior techniques, does not need to solve a linear program.
Despite this, it allows for deriving that the well-known hypergraphic Steiner tree relaxation has an integrality gap of no more than $\ln 4$.
This has been shown previously in~\cite{goemans_2012_matroids}; however, our proof is arguably simpler than the one presented in~\cite{goemans_2012_matroids}, which is based on building up a thorough understanding of how a highly fractional LP solution can be modified iteratively. 
Furthermore, the analysis of our approach is based on a classic exchange property about spanning trees (or, more generally, matroids).
In this way, we achieve the approximation ratio $\ln 4 + \epsilon$ without relying on the bride lemma, which was a key technical component of the analysis in~\cite{byrka_2013_steiner}.

\subsection{Organization of the paper}

We start by showing in Section~\ref{sec:WTAP} our main result, Theorem~\ref{thm:mainWTAP}.
This allows us to showcase our approach and its advantages compared to the previously strongest results, which is based on a relative greedy approach.
In Section~\ref{sec:steinerTree}, we show how our approach can be extended to the Steiner tree problem.
Finally, we discuss crucial differences between our approach and iterative randomized rounding, which led to the first $(\ln 4 + \epsilon)$-approximation for the Steiner tree problem.
Despite the fact that both approaches lead to the same approximation guarantee for the Steiner tree problem, there are significant barriers to apply iterative randomized rounding, or similar approaches, in the context of WTAP.
The reason is that these approaches require a stronger decomposition result.
We expand on this in Section~\ref{sec:comparison}, which helps to develop a better understanding of how our approach relates to prior techniques.
\section{A $(1.5+\epsilon)$-approximation for WTAP}\label{sec:WTAP}

In this section we present our local search algorithm for WTAP and prove our main result, Theorem~\ref{thm:mainWTAP}.
After introducing some basic terminology in Section~\ref{sec:preliminaries_wtap}, we first recap the relative greedy algorithm from~\cite{traub_2021_better} (Section~\ref{sec:recap_relative_greedy}).
This allows us to discuss some results that will be reused in our approach, and helps to understand how we improve on prior approaches.
We then give an overview of our new local search algorithm (Section~\ref{sec:new_algo_outline}).
Finally, we define the potential function that we use to measure progress in our algorithm (Section~\ref{sec:potential}) and describe the details of the algorithm and its analysis (Section~\ref{sec:details_algo_wtap}).

\subsection{Preliminaries}\label{sec:preliminaries_wtap}

Recall that an instance of WTAP consists of a spanning tree $G=(V,E)$ and a set $L\subseteq \left(\!\begin{smallmatrix} V\\ 2 \end{smallmatrix}\!\right)$ of links with weights $w\colon L \to \mathbb{R}_{> 0}$.
For a link $\ell=\{a,b\}\in L$, we denote by $P_{\ell}\subseteq E$ the set of edges that are contained on the unique $a$-$b$ path in $G$.
It is well-known and easy to see that a set $F\subseteq L$ is a WTAP solution if and only if every edge of the tree $G$ is contained in one of the paths $P_{\ell}$ with $\ell\in F$.
Thus, we can naturally view WTAP as a covering problem, where we want to cover the edge set $E$ of the tree $G$ by links.

It is often useful to assume that the given WTAP instance is a so-called \emph{shadow complete} instance.
A \emph{shadow} of a link $\ell =\{a,b\}\in L$ is a link $\overline{\ell}$ with $P_{\overline{\ell}} \subseteq P_{\ell}$, or equivalently, a link $\overline{\ell}=\{\overline{a},\overline{b}\}$ such that both endpoints $\overline{a}$ and $\overline{b}$ lie on the $a$-$b$ path in $G$.
We can always assume that, for each link $\ell\in L$, the set $L$ contains also all shadows of $\ell$ and that the weight of each shadow of $\ell$ is no larger than the weight of $\ell$,
in which case we call the WTAP instance \emph{shadow complete}.
This assumption is without loss of generality because we can add for every link $\ell\in L$ all its shadows and give weight $w(\ell)$ to each of them; then, in any WTAP solution $F$, we can replace any such shadow of $\ell$ by the link $\ell$ itself, obtaining a solution of same weight.

\subsection{Recap of the Relative Greedy Algorithm for WTAP}\label{sec:recap_relative_greedy}

In this section we recall the relative greedy algorithm from~\cite{traub_2021_better}.
In this algorithm, we first fix an arbitrary root $r$ of the tree $G$ and compute a WTAP solution consisting only of up-links.
An \emph{up-link} is a link $\ell$ where one of the endpoints of the link is an ancestor of the other endpoint, i.e., one endpoint of $\ell$ lies on the unique path in the tree $G$ from the root to the other endpoint of $\ell$. 
In the following, we denote the set of up-links by $L_{\mathrm{up}}\subseteq L$.
The WTAP solution $U \subseteq L_{\mathrm{up}}$ that the relative greedy algorithm starts with is a $2$-approximation, i.e., we have $w(U) \le 2 \cdot w(\OPT)$, where $\OPT$ denotes an optimal WTAP solution.

The well-known lemma below (see, e.g.,~\cite{cohen_2013_approximation}) allows us to assume without loss of generality that the paths $P_u$ with $u\in U$ are pairwise disjoint.
\begin{lemma}\label{lem:shorten_up_links}%
\footnote{
The lemma readily follows by the following two-step procedure to modify $U$.
First, successively delete redundant links in $U$ to obtain $U'$, i.e., links that can be deleted while maintaining a solution.
Second, each link $u\in U'$, in an arbitrary order, gets shortened to the smallest shadow $\overline{u}$ of $u$ for which the shortening maintains a solution.
The resulting link set $U''$ is such that the paths $P_u$ for $u\in U''$ are disjoint; indeed, if we had $P_{u_1}\cap P_{u_2}\neq \emptyset$ for distinct $u_1,u_2\in U''$, then $u_1$ or $u_2$ could have been further shortened.
}
Let $(G=(V,E),L,w)$ be a shadow-complete WTAP instance, and let $U\subseteq L_{\mathrm{up}}$ be a WTAP solution.
Then we can efficiently transform $U$ into a WTAP solution for which the paths $P_u$ with $u\in U$ are disjoint by replacing some links $u\in U$ by one of its shadows and possibly removing some links from $U$.
\end{lemma}

In the relative greedy algorithm, we then improve the solution $U\subseteq L_{\mathrm{up}}$ as follows. 
We determine a well-chosen link set $C\subseteq L$, which we also call a component.
Then we add $C$ to the current WTAP solution and remove all up-links from $U\subseteq L_{\mathrm{up}}$ that become redundant, i.e., we remove all links in
\begin{equation*}
\Drop_U(C) \coloneqq \left\{u\in U\colon P_u \subseteq \bigcup_{\ell\in C} P_{\ell}\right\}\enspace.
\end{equation*}
This gets iterated until all up-links from the initial WTAP solution are replaced.
The links in the newly added components $C$ are not necessarily up-links and will never be removed.
The key difficulty lies in efficiently finding a good component $C$.
To this end, \textcite{traub_2021_better} restrict the search space to components $C\subseteq L$ that are \emph{$k$-thin}, which means that, for every vertex $v\in V$, there are at most $k$ links $\{a,b\}\in C$ for which $v$ lies on the $a$-$b$ path in $G$.

A key part of the analysis of the relative greedy algorithm is to show that there always exists a good $k$-thin component to add next.
This is achieved in~\cite{traub_2021_better} through a decomposition theorem (Theorem~\ref{thm:decomposition} below), which implies that for any set $U\subseteq L_{\mathrm{up}}$ of up-links for which the paths $P_u$, for $u\in U$, are disjoint, there exists a partition $\Cscr$ of the optimum solution $\OPT$ into $k$-thin components such that the following holds.
If we sample a component $C\in\Cscr$ uniformly at random, we have $\mathbb{E}[w(\Drop_U(C))] \ge (1-\epsilon) \frac{1}{|\Cscr|} w(U)$, where $\epsilon>0$ is an arbitrarily small fixed number.
Because $\mathbb{E}[w(C)] = \frac{1}{|\Cscr|} w(\OPT)$, this shows that as long as $w(U)$ is significantly larger than $w(\OPT)$, there is a $k$-thin component that we can use to improve our WTAP solution.

\begin{restatable}[Theorem~5 in \cite{traub_2021_better}]{theorem}{decompositionTheorem}\label{thm:decomposition}
Let $(G=(V,E),L,w)$ be a WTAP instance, $F\subseteq L$ be a WTAP solution, and let $U\subseteq L_{\mathrm{up}}$ be a set of up-links such that the sets $P_u$ with $u\in U$ are pairwise disjoint.
Then, for any $\epsilon >0$, there exists a partition $\Cscr$ of $F$ into $\lceil\sfrac{1}{\epsilon}\rceil$-thin sets and a set $R\subseteq U$ such that
\begin{enumerate}
\item for every $u\in U\setminus R$, there exists some $C \in \Cscr$ such that $P_u \subseteq \bigcup_{\ell\in C} P_{\ell}$, and  \label{item:links_outside_R_covered}
\item $w(R) \le \epsilon \cdot w(U)$.
\end{enumerate}
\end{restatable}

\subsection{Improving the Approximation Guarantee through Local Search}\label{sec:new_algo_outline}

We now discuss how we improve on the algorithm from~\cite{traub_2021_better}. 
Instead of removing only redundant up-links that were part of the original $2$-approximation we started with, we also want to make progress by dropping links from components we added in previous iterations.
One reason why one might hope to obtain better solutions through such an approach is that in the analysis of the relative greedy algorithm we only used that the components in the partition $\Cscr$ of $\OPT$ cover the paths $P_u$ for the up-links $u\in U$.
But these components have the additional property that they cover all edges of the tree $G$, including those covered by the components added in earlier iterations.

However, this fact alone is not sufficient to obtain an improved approximation guarantee because covering the edges in $P_{\ell}$ for a link $\ell$ that was selected in a previous iteration could require many different components from $\Cscr$ and we can only remove $\ell$ once all edges in $P_{\ell}$ are covered by other links than $\ell$. In order to handle this, we use the simple and well-known observation that we can split every link $\ell\in L$ into (at most) two up-links that cover the same edges as $\ell$. Formally, for a link $\ell=\{a,b\}\in L\setminus L_{\mathrm{up}}$, we define 
\[
U_{\ell} \coloneqq  \{ \{a, \apex(\ell)\}, \{b,\apex(\ell)\} \} \subseteq L_{\mathrm{up}}\enspace,
\]
where $\apex(\ell)$ is the common ancestor of $a$ and $b$ in the tree $G$ that is farthest away from the root $r$.
For an up-link $\ell \in L_{\mathrm{up}}$ we define $U_{\ell}\coloneqq \{\ell\}$.
Then $U_{\ell}$ is indeed a set of at most two up-links that cover the same edges as the link $\ell$, i.e., $P_{\ell} = \bigcup_{u\in U_{\ell}} P_u$, and all links in $U_{\ell}$ are shadows of $\ell$.
Thus, as soon as we added components that cover all up-links in $U_{\ell}$, we can safely remove the link $\ell$ from our current WTAP solution.

In the analysis of our new algorithm we will apply the decomposition theorem from~\cite{traub_2021_better} not just to the set $U$ of up-links from the $2$-approximation we started with, but to the union of $U$ and the sets $U_{\ell}$ for all links $\ell$ in our current WTAP solution, which includes the components we added in previous iterations. 
(More precisely, we first apply Lemma~\ref{lem:shorten_up_links} and then Theorem~\ref{thm:decomposition} to this set of up-links to ensure that the assumptions of Theorem~\ref{thm:decomposition} are fulfilled.)
Indeed, if we now iteratively selected a random component $C$ in the resulting partition $\Cscr$ of $\OPT$, then this would lead to an improved approximation algorithm because we can remove a link $\ell$ as soon as the (at most) two up-links in $U_{\ell}$ are covered.
As each of these up-links is completely covered by a single component, intuitively there is a significant probability that this happens at some point during the  algorithm.

However, this algorithm is not feasible because without knowing $\OPT$ we cannot construct the partition $\Cscr$ of $\OPT$ into $k$-thin components.
Therefore, we define a suitable potential function $\Phi$, which we will formally define in Section~\ref{sec:potential}, and we will aim at finding components that lead to a large decrease of this potential $\Phi$.
The potential $\Phi$ (which is inspired by prior work~\cite{goemans_2012_matroids} in the context of the Steiner tree problem) also rewards partial progress, i.e., when a component is selected that only covers one of two links of some set $U_{\ell}$.
Using a dynamic programming algorithm from~\cite{traub_2021_better} we are able to find a $k$-thin component $C$ for which the decrease of the potential $\Phi$ is at least as large as the potential decrease that we could guarantee when we contracted a random component from the partition $\Cscr$ of an optimum solution.

It turns out that, with this improved algorithm, it is no longer necessary that the WTAP solution we start with is a $2$-approximation using only up-links.
Thus, our algorithm can naturally be described as a local search procedure that starts with an arbitrary solution $F$ and iteratively tries to find a $k$-thin component that can be used to decrease the potential $\Phi(F)$. As soon as we cannot anymore find a component that leads to a (significant) decrease of $\Phi$, our algorithm returns the current solution $F$.

Throughout the course of our algorithm we maintain
\begin{itemize}
\item a WTAP solution $F$, and
\item for every link $\ell=\{a,b\}\in F$, a non-empty set $W_{\ell} \subseteq L_{\mathrm{up}}$ of at most two shadows of $\ell$, called the \emph{witness set} of $\ell$, which we initially set to $U_{\ell} \coloneqq  \{ \{a, \apex(\ell)\}, \{b,\apex(\ell)\} \}$. 
\end{itemize}
We will always have the property that the disjoint union $U=\bigcupp_{\ell\in F} W_{\ell}$ of the witness sets is a WTAP solution and we will use Lemma~\ref{lem:shorten_up_links} to ensure that the paths $P_u$ with $u\in U$ are disjoint.
Whenever we add a new component $C\subseteq L$ to the solution $F$, we will remove all up-links in $\Drop_U(C)$ from  $U$ and from all witness sets. 
Moreover, for all $\ell\in C$, we add the links in the witness sets $W_{\ell}=U_{\ell}$ to $U$.
As soon as the witness set $W_{\ell}$ of a link $\ell\in F$ becomes empty, we remove the link $\ell$ from the WTAP solution $F$.
See Figure~\ref{fig:example_exchange} for an example.

\begin{figure}
\begin{center}
\begin{tikzpicture}[xscale=0.8,yscale=1.1]

\tikzset{
  prefix node name/.style={%
    /tikz/name/.append style={%
      /tikz/alias={#1##1}%
    }%
  }
}

\tikzset{
lks/.style={line width=2pt, densely dashed},
}

\newcommand\tree[2][]{

\begin{scope}[prefix node name=#1]

\begin{scope}[every node/.style={thick,draw=black,fill=black,circle,minimum size=0pt, inner sep=1.2pt, outer sep=1pt}]

\node (0) at (3,3.5) {};
\node (1) at (2,3) {};
\node (2) at (5.5,3) {};
\node (3) at (4,3) {};
\node (4) at (5,2.5) {};
\node (5) at (4,2) {};
\node (6) at (3,1.5) {};
\node (7) at (6,2) {};
\node (9) at (5,1.5) {};
\node (10) at (6,1) {};
\node (11) at (7,1.5) {};
\node (12) at (7,0.5) {};
\node (13) at (8,0.5) {};

\end{scope}

\node[above] (r) at (0) {$r$};

\begin{scope}[very thick]
\draw (1) -- (0);
\draw (3) -- (2);
\draw (0) -- (3) -- (4) -- (5) -- (6);
\draw (4) -- (7);
\draw (7) -- (9);
\draw (7) -- (10);
\draw (7) -- (11) -- (12);
\draw (11) -- (13);
\end{scope}

\end{scope}

}%

\begin{scope}[shift={(0,0)}]
\tree[f1-]{}
\node (f1) at (-1,2) {
\begin{minipage}{3.5cm}
Solutions $F$ and $U$ \\ 
before adding $C$.
\end{minipage}};
\end{scope}

\begin{scope}[shift={(8,0)}]
\tree[u1-]{}
\end{scope}

\begin{scope}[shift={(0,-5)}]
\tree[c-]{}
\node (c) at (-1,2) {
\begin{minipage}{3.5cm}
Component $C$ \\
and up-links $\bigcupp_{\ell \in C} U_{\ell}$.
\end{minipage}};
\end{scope}

\begin{scope}[shift={(8,-5)}]
\tree[cu-]{}
\end{scope}

\begin{scope}[shift={(0,-10)}]
\tree[f2-]{}
\node (f2) at (-1,2) {
\begin{minipage}{3.5cm}
Solutions $F$ and $U$ \\ 
after adding $C$.
\end{minipage}};
\end{scope}

\begin{scope}[shift={(8,-10)}]
\tree[u2-]{}
\end{scope}

\begin{scope}[lks]

\begin{scope}[orange!90!black]%
\draw (f1-1) to (f1-5);
\draw (f2-1) to (f2-5);
\draw[bend left] (u1-1) to (u1-0);
\draw (u1-5) to (u1-0);
\draw[bend left] (u2-1) to (u2-0);
\draw (u2-5) to (u2-0);
\end{scope}

\begin{scope}[violet]%
\draw (f1-3) to (f1-6);
\draw[bend left] (u1-6) to (u1-5);
\end{scope}

\begin{scope}[darkred]%
\draw (f1-9) to (f1-10);
\draw[bend left] (u1-9) to (u1-7);
\draw[bend left] (u1-7) to (u1-10);
\end{scope}

\begin{scope}[darkgreen]%
\draw (f1-2) to (f1-11);
\draw (f2-2) to (f2-11);
\draw[bend right] (u1-2) to (u1-3);
\draw[bend right] (u1-11) to (u1-4);
\draw[bend right] (u2-2) to (u2-3);
\end{scope}

\begin{scope}[yellow!60!black]%
\draw (f1-12) to (f1-13);
\draw (f2-12) to (f2-13);
\draw[bend left] (u1-12) to (u1-11);
\draw[bend left] (u1-11) to (u1-13);
\draw[bend left] (u2-12) to (u2-11);
\draw[bend left] (u2-11) to (u2-13);
\end{scope}

\begin{scope}[cyan]%
\draw (c-6) to (c-9);
\draw (f2-6) to (f2-9);
\draw[bend left] (cu-6) to (cu-4);
\draw (cu-4) to (cu-9);
\draw[bend left] (u2-6) to (u2-5);
\draw (u2-4) to (u2-9);
\end{scope}

\begin{scope}[blue] %
\draw (c-11) to (c-10);
\draw (f2-11) to (f2-10);
\draw[bend right] (cu-11) to (cu-7);
\draw[bend left] (cu-10) to (cu-7);
\draw[bend right] (u2-11) to (u2-7);
\draw[bend left] (u2-10) to (u2-7);
\end{scope}

\end{scope} %

\end{tikzpicture}
 \end{center}
\caption{Example of a local improvement step. The left column shows (from top to bottom) a WTAP solution $F$, a component $C$, and the new solution $F$ after adding $C$ through a local exchange step.
The right column shows the witness sets of the links shown in the left coloumn, i.e., the top and bottom right picture show the
 WTAP solution $U =\ \cupp_{\ell \in F} W_{\ell} \subseteq L_{\mathrm{up}}$ before and after the exchange step adding the component $C$, and the middle right picture shows $\cupp_{\ell\in C}U_{\ell}$. The witness set of a link is drawn in the same color as the link itself.
Note that in the WTAP solution $U$ we shorten up-links, i.e., replace them by a shadow, to ensure that the paths $P_u$ with $u\in U$ are disjoint (Lemma~\ref{lem:shorten_up_links}).
In this example, the violet and red witness set become empty when we remove $\Drop_U(C)$ from $U$ and, hence, we also remove the violet and red links in $F$. Moreover, the size of the dark green witness set decreases from $2$ to $1$. This change does not affect the solution $F$, but it will lead to a decrease of the potential $\Phi(F)$.
\label{fig:example_exchange}}
\end{figure}

\subsection{The Potential Function $\Phi$}\label{sec:potential}

In this section we define the potential function $\Phi$ that we use to measure progress in our local search procedure.
First, we assign positive weights $\overline{w}$ to the up-links in $U=\cupp_{\ell\in F} W_{\ell}$, where we distribute the weight $w(\ell)$ of a link $\ell \in F$ equally among the up-links in its witness set.
Formally, for an up-link $u\in W_{\ell}$, we define 
\begin{equation*}
\overline{w}(u)\ \coloneqq\ \frac{w(\ell)}{|W_{\ell}|}\enspace.
\end{equation*}
Because we simply spread the weight of $w(F)$, the total $\overline{w}$-weight is equal to $w(F)$:
\begin{equation}\label{eq:spread_weight_wtap}
\overline{w}(U)\ =\ \sum_{\ell\in F} \sum_{u\in W_{\ell}} \overline{w}(u)\ =\ \sum_{\ell\in F} \sum_{u\in W_{\ell}} \frac{w(\ell)}{|W_{\ell}|} \ =\ \sum_{\ell\in F} w(\ell)\ =\ w(F)\enspace.
\end{equation}
We will define the potential $\Phi$ such that it fulfills the following key properties:
\begin{enumerate}[label=(\alph*)]
\item\label{item:phi_decrease_wtap} $\Phi(F)$ decreases by at least $\overline{w}(\Drop_U(C))$ when we remove $\Drop_U(C)$ from all witness sets $W_{\ell}$ with $\ell\in F$ and remove the links with empty witness sets from $F$, and
\item\label{item:phi_increase_wtap} $\Phi(F)$ increases by at most $1.5\cdot w(C)$ when we add $C$ to $F$.
\end{enumerate}
To see why these properties lead to the desired approximation guarantee of our local search algorithm, consider the partition $\Cscr$ of $\OPT$ into $k$-thin components that we obtain from Theorem~\ref{thm:decomposition}, applied to the WTAP solution $U\subseteq L_{\mathrm{up}}$ with weights $\overline{w}$. 
If we sample one of these components uniformly at random, the expected decrease of $\Phi(F)$ by removing $\Drop_U(C)$ from all witness sets is at least
\[
\frac{1}{|\Cscr|}\cdot \sum_{C\in \Cscr} \overline{w}(\Drop_U(C)) \ge \frac{1-\epsilon}{|\Cscr|} \cdot \overline{w}(U) = \frac{1-\epsilon}{|\Cscr|} \cdot w(F)\enspace,
\] 
(by \ref{item:phi_decrease_wtap}, Theorem~\ref{thm:decomposition}, and \eqref{eq:spread_weight_wtap}) while the expected increase of $\Phi(F)$ when adding $C$ is at most $\frac{1}{|\Cscr|} \cdot 1.5 \cdot w(\OPT)$ (by \ref{item:phi_increase_wtap}).
Thus, as long as $w(F)$ is significantly larger than $1.5\cdot w(C)$, the potential $\Phi(F)$ decreases in expectation.
This argument shows that as long as the solution $F$ does not fulfill the desired upper bound on its weight, there exists a $k$-thin component that we can use to decrease $\Phi(F)$.
In order to find such a component efficiently, we use a dynamic programming algorithm from~\cite{traub_2021_better}, which yields the following.

\begin{lemma}\label{lem:find_improvement}
Let $k\in \mathbb{Z}_{\ge 0}$ be a constant.
Given a WTAP instance $(G=(V,E),L,w)$, a set $U\subseteq L_{\mathrm{up}}$ of up-links such that the sets $P_u$ with $u\in U$ are pairwise disjoint, and weights $\overline{w}(u)>0$ for all $u\in U$,
we can efficiently compute a $k$-thin link set $C\subseteq L$ maximizing
$\overline{w}(\Drop_U(C)) - 1.5 \cdot w(C)$.  
\end{lemma}
\begin{proof}
We consider the WTAP instance $(G=(V,E),L,\widetilde w)$, where 
\begin{equation*}
\widetilde w(\ell) \coloneqq 
\begin{cases}
\overline{w}(\ell) &\text{ if }\ell\in U, \\
1.5\cdot w(\ell)&\text{ otherwise}.
\end{cases}
\end{equation*}
Applying Lemma~17 from \cite{traub_2021_better} to this instance (with $\rho=1$) completes the proof.
\end{proof}

Let us now define the potential function $\Phi$.
For a set $F\subseteq L$ with witness sets $W_{\ell}$ for $\ell\in F$, we define the potential 
\[
\Phi(F) \ \coloneqq\ \sum_{\ell\in F} H_{|W_{\ell}|} \cdot w(\ell)\ =\ \sum_{\ell \in F: |W_{\ell}|=1} w(\ell) +\sum_{\ell \in F: |W_{\ell}|=2} \frac{3}{2} \cdot w(\ell)\enspace,
\]
where $H_i\coloneqq \sum_{j=1}^i \frac{1}{j}$ for $i\in \mathbb{Z}_\geq 0$, and we used for the last equality that we always have $|W_{\ell}|\in \{1,2\}$.
Then the potential $\Phi$ satisfies \ref{item:phi_decrease_wtap} and \ref{item:phi_increase_wtap}; see Lemma~\ref{lem:potential_decrease_wtap} below.

\subsection{The Local Search Algorithm}\label{sec:details_algo_wtap}

Our local search algorithm for WTAP can now be stated as follows, where we fix a constant $0 < \epsilon \le \sfrac{1}{2}$.

\begin{algorithm2e}[H]
\KwIn{A shadow-complete WTAP instance $(G=(V,E),L,w)$.}
\KwOut{A WTAP solution $F\subseteq L$ with $w(F) \le (1.5+\epsilon)\cdot w(\OPT)$.}
\vspace*{2mm}
\begin{enumerate}[label=\arabic*.,ref=\arabic*,rightmargin=7mm]\itemsep4pt
\item\label{item:initialize} Let $F\subseteq L$ be an arbitrary solution for the given WTAP instance. \\
  Set the witness sets to be $W_{\ell} \coloneqq U_{\ell}$ for all $\ell\in F$ and apply Lemma~\ref{lem:shorten_up_links} to $U= \bigcupp_{\ell \in F} W_{\ell}$.
\item\label{item:local_step_wtap} Iterate the following as long as $\Phi(F)$ decreases in each iteration by at least a factor $\left( 1- \frac{\epsilon}{6\cdot|V|}\right)$.
\begin{itemize}\itemsep1pt
\item \textbf{Select a best component:} 
Compute a $\lceil \sfrac{4}{\epsilon} \rceil$-thin link set $C\subseteq L$ maximizing
$\overline{w}(\Drop_U(C)) - 1.5 \cdot w(C)$, where $U=\ \bigcupp_{\ell \in F} W_{\ell}$.
\item \textbf{Remove $\mathrm{\mathbf{Drop}}$:} Replace the witness set $W_{\ell}$ by $W_{\ell} \setminus \Drop_U(C)$ for all $\ell \in F$.
\item \textbf{Add the new component:} Add $C$ to $F$ and set $W_{\ell} \coloneqq U_{\ell}$ for all $\ell \in C$.
\item \textbf{Shorten up-links:} Apply Lemma~\ref{lem:shorten_up_links} to $U= \bigcupp_{\ell \in F} W_{\ell}$.
\item If for some link $\ell \in F$, the witness set $W_{\ell}$ became empty, remove $\ell$ from $F$.
\end{itemize}
\item Return $F$.
\end{enumerate}
\caption{Local search algorithm for WTAP}\label{algo:local_search}
\end{algorithm2e}

The applications of Lemma~\ref{lem:shorten_up_links} in both step~\ref{item:initialize} and in the ``shorten up-links'' operation in step~\ref{item:local_step_wtap} are to be interpreted as follows. The link set $U$, which is a WTAP solution (see Lemma~\ref{lem:algReturnsWTAPSol} below), gets replaced by a shortened up-link solution by removing and shortening links in $U$ (as stated in Lemma~\ref{lem:shorten_up_links}). When shortening a link $u$ to one of its shadows $u'$, then a witness set $W_\ell$ that used to contain $u$ will now contain $u'$ instead (as mentioned, we think of $u'$ as replacing the up-link $u$).

We first show that our local search algorithm returns a feasible solution.

\begin{lemma}\label{lem:algReturnsWTAPSol}
Both $F$ and $U=\bigcupp_{\ell \in F} W_{\ell}$ are WTAP solutions before and after each iteration of Algorithm~\ref{algo:local_search}.
In particular, when the algorithm terminates, it returns a WTAP solution.
\end{lemma}
\begin{proof}
In step~\ref{item:initialize} of Algorithm~\ref{algo:local_search}, we set $F\subseteq L$ to be a WTAP solution.
Because we have $P_{\ell} =\bigcup_{u \in U_{\ell}}P_u$ for every link $\ell \in F$, this implies that also $U$ is a WTAP solution after step~\ref{item:initialize}.

When we add a component $C$ to $F$, we set $W_{\ell}=U_{\ell}$ and add the up-links in $U_C \coloneqq\cupp_{\ell \in C} U_{\ell}$ to $U$. 
Because these up-links cover the same edges as the links in the component $C$, we have $\Drop_U(U_C) = \Drop_U(C)$ and thus $U$ remains a WTAP solution in step~\ref{item:local_step_wtap}.
In order to show that also $F$ remains a WTAP solution, we observe that we maintain the invariant that the elements of the witness set $W_{\ell}$ of a link $\ell$ are shadows of $\ell$.
Because $U$ is a WTAP solution, every edge $e$ of the tree $G$ is covered by some up-link $u\in U$ that is contained in the witness set $W_{\ell}$ for some link $\ell \in F$. 
Hence, because $u$ is a shadow of $\ell$, we can conclude that also $\ell$ covers the edge $e$.
This shows that not only $U$, but also $F$ remains a WTAP solution.
\end{proof}

To prove that the solution returned by Algorithm~\ref{algo:local_search} fulfills the desired approximation guarantee, we use the following observation, which follows from the definition of the potential function $\Phi$.

\begin{lemma}\label{lem:potential_decrease_wtap}
If we select a component $C\subseteq L$ in step~\ref{item:local_step_wtap} of Algorithm~\ref{algo:local_search}, then $\Phi(F)$ decreases by at least $\overline{w}(\Drop_U(C)) - 1.5 \cdot w(C)$ in this iteration.
\end{lemma}
\begin{proof}
Adding $C$ to $F$ increases the potential $\Phi(F)$ by 
$\sum_{\ell\in C} H_{|U_{\ell}|} \cdot w(\ell) \le \sum_{\ell\in C} \frac{3}{2}\cdot w(\ell)$ because $|U_{\ell}|\le 2$ for all $\ell\in L$.
Thus, it remains to show that $\Phi(F)=\sum_{\ell\in F} H_{|W_{\ell}|} \cdot w(\ell)$ decreases by at least $\overline{w}(\Drop_U(C))$ when we replace $W_{\ell}$ by $W_{\ell} \setminus \Drop_U(C)$ for all $\ell\in F$.
To this end, we consider a link $\ell\in F$ and show that $H_{|W_{\ell}|} \cdot w(\ell)$ decreases by at least $\overline{w}(W_{\ell} \cap \Drop_U(C))$.
If there is exactly one link $u\in W_{\ell} \cap \Drop_U(C)$ , then $H_{|W_{\ell}|}\cdot w(\ell)$ decreases by $\frac{1}{|W_{\ell}|} \cdot w(\ell) =\overline{w}(u)$ when we replace $W_{\ell}$ by $W_{\ell} \setminus \Drop_U(C)$.
If there are two links $u_1,u_2 \in W_{\ell} \cap \Drop_U(C)$, then $W_{\ell}=\{u_1,u_2\}$. 
Hence in this case we have $\overline{w}(u_1) + \overline{w}(u_2) = w(\ell)$ and $H_{|W_{\ell}|}\cdot w(\ell)$ decreases from $\frac{3}{2}\cdot w(\ell)$ to $0$ when we replace $W_{\ell}$ by $W_{\ell} \setminus \Drop_U(C)$.
\end{proof}

Together with Lemma~\ref{lem:potential_decrease_wtap}, the lemma below gives a lower bound on the decrease of the potential $\Phi(F)$ in a single local improvement step.
This lower bound is positive as long as $w(F)$ is significantly larger than $1.5\cdot w(\OPT)$.
To prove Lemma~\ref{lem:good_improvement_exists_wtap}, we give a lower bound on the improvement that can be achieved through the components obtained from the decomposition theorem (Theorem~\ref{thm:decomposition}) applied to $\OPT$.

\begin{lemma}\label{lem:good_improvement_exists_wtap}
In every iteration of Algorithm~\ref{algo:local_search}, there exists a $\lceil\sfrac{4}{\epsilon}\rceil$-thin component $C\subseteq L$ such that
\begin{equation}\label{eq:wtap_bound_improvement}
\overline{w}(\Drop_U(C)) - 1.5 \cdot w(C)\ \ge\ \frac{1}{|V|} \cdot \Big((1-\sfrac{\epsilon}{4}) \cdot w(F) - 1.5 \cdot w(\OPT)\Big)\enspace.
\end{equation}
\end{lemma}
\begin{proof}
We apply Theorem~\ref{thm:decomposition} to $U$ with weight function $\overline{w}$ to obtain a partition $\Cscr$ of $\OPT$ into $\lceil\sfrac{4}{\epsilon}\rceil$-thin components such that 
\[
\sum_{C\in \Cscr} \overline{w}(\Drop_U(C))\ \ge\ (1-\sfrac{\epsilon}{4}) \cdot \overline{w}(U)\ =\ (1-\sfrac{\epsilon}{4}) \cdot w(F)\enspace.
\]
Because $\sum_{C\in \Cscr} w(C) = w(\OPT)$, we obtain 
\begin{align*}
\max_{C\in\Cscr} \left(\overline{w}(\Drop_U(C))-1.5 \cdot w(C)\right)\ \ge&\ \frac{1}{|\Cscr|}\sum_{C\in \Cscr} \left(\overline{w}(\Drop_U(C))-1.5\cdot w(C)\right)\\
 \ge&\ \frac{1}{|V|}\cdot\left((1-\sfrac{\epsilon}{4}) \cdot w(F) - 1.5\cdot w(\OPT)\right)\enspace,
\end{align*}
where we used $|\Cscr| \le |V|$.
\end{proof}

Next, we bound the weight of the solution returned by Algorithm~\ref{algo:local_search}.

\begin{lemma}\label{lem:wtap_apx_guarantee}
When Algorithm~\ref{algo:local_search} terminates, it returns a WTAP solution $F$ with $w(F)\le (1.5+\epsilon)\cdot w(\OPT)$.
\end{lemma}
\begin{proof}
Lemma~\ref{lem:potential_decrease_wtap} implies that when the algorithm terminates, we must have $\overline{w}(\Drop_U(C)) -  1.5 \cdot w(C) < \frac{\epsilon}{6\cdot |V|} \cdot \Phi(F)$ for every $\lceil\sfrac{4}{\epsilon}\rceil$-thin component $C\subseteq L$.
By Lemma~\ref{lem:good_improvement_exists_wtap}, this implies
\[
(1-\sfrac{\epsilon}{4}) \cdot w(F) - 1.5 \cdot w(\OPT)
\ <\ \frac{\epsilon}{6} \cdot \Phi(F) \ \le\ \frac{\epsilon}{4} \cdot w(F)\enspace,
\]
where we used $\Phi(F)\le \sfrac{3}{2} \cdot w(F)$.
Therefore, $(1-\sfrac{\epsilon}{2})\cdot w(F)\le 1.5 \cdot w(\OPT)$ and thus $w(F) \le (1.5+\epsilon)\cdot w(\OPT)$, using $\epsilon \le \sfrac{1}{2}$.
\end{proof}

Finally, we show that our local search procedure terminates in polynomial time. Note that the starting WTAP solution $F_0$, computed in step~\ref{item:initialize} of Algorithm~\ref{algo:local_search}, has weight bounded by $w(F_0)\leq w(L)$; thus, the bound stated in the lemma below is indeed polynomial, independently of the starting solution $F_0$.
\begin{lemma}\label{lem:bound_iter_WTAP}
Algorithm~\ref{algo:local_search} terminates after at most $\ln\!\left(\frac{\sfrac{3}{2}\cdot w(F_0)}{w(\OPT)}\right)\cdot \frac{6|V|}{\epsilon}$ iterations, where $F_0\subseteq L$ is the initial WTAP solution computed in step~\ref{item:initialize} of Algorithm~\ref{algo:local_search}.
\end{lemma}
\begin{proof}
At the beginning of the local search algorithm we have $\Phi(F)=\Phi(F_0)\le \frac{3}{2} \cdot w(F_0)$.
Because the potential $\Phi(F)$ decreases by a factor of at least $\left(1- \frac{\epsilon}{6\cdot |V|}\right)$ in every iteration and because $\Phi(F)\ge w(F)\ge w(\OPT)$ throughout the algorithm, the number of iterations is at most
\begin{equation*}
\log_{(1-\sfrac{\epsilon}{(6|V|)})^{-1}} \left(\frac{\sfrac{3}{2}\cdot w(F_0)}{w(\OPT)}\right) \ =\ \ln\left(\frac{\sfrac{3}{2}\cdot w(F_0)}{w(\OPT)}\right)\cdot \frac{1}{-\ln(1-\sfrac{\epsilon}{(6|V|)})}\ \le\ \ln\left(\frac{\sfrac{3}{2}\cdot w(F_0)}{w(\OPT)}\right)\cdot \frac{6|V|}{\epsilon}\enspace,
\end{equation*}
where we used $\ln(1+x) \le x$ for $x> -1$.
\end{proof}

Combining Lemma~\ref{lem:wtap_apx_guarantee} and Lemma~\ref{lem:bound_iter_WTAP} yields that Algorithm~\ref{algo:local_search} is a $(1.5+\epsilon)$-approximation algorithm for WTAP and thus completes the proof of Theorem~\ref{thm:mainWTAP}.

\section{Local Search for Steiner Tree}\label{sec:steinerTree}

In this section we discuss our local search algorithm for the Steiner tree problem.
An instance of the Steiner tree problem consists of an undirected graph $G=(V,E)$ with positive edge weights $w : E \to \mathbb{R}_{> 0}$ and a set $T\subseteq V$ of \emph{terminals}.
The task is to find a set $F$ of edges that connects all terminals, i.e., an edge set $F$ such that the graph $(V,F)$ contains a path between any pair of terminals.\footnote
{
Sometimes the Steiner tree problem is defined such that edges of weight zero are allowed to exist.
However, this is equivalent because edges of weight zero can always be included in any solution at no extra cost and thus contracting these edges upfront yields an equivalent instance with positive weights only.
}

The currently best approximation algorithm for the Steiner tree problem is an $(\ln 4 + \epsilon)$-approximation algorithm by \textcite{byrka_2013_steiner} through an elegant iterative randomized rounding method.
After some preliminaries in Section~\ref{sec:steiner_preliminaries}, we describe a new local search algorithm and prove that it achieves the same approximation ratio without the need to solve a linear program (Section~\ref{sec:steiner_algo}).

\textcite{goemans_2012_matroids} gave a variant of the $(\ln 4 + \epsilon)$-approximation algorithm from~\cite{byrka_2013_steiner} and proved that the cost of the output of their algorithm can be bounded with respect to the optimal value of the well-known hypergraphic LP relaxation for Steiner tree (see Section~\ref{sec:steiner_lp}).
In particular, they proved that the hypergraphic LP relaxation  has an integrality gap of at most $\ln 4$.
In Section~\ref{sec:steiner_lp} we give a simpler proof of this result by showing that our local search procedure computes a solution of cost no more than $\ln 4 + \epsilon$ times the LP value.

\subsection{Components and $k$-restricted Steiner trees}\label{sec:steiner_preliminaries}

A \emph{component} is a nonempty edge set $C\subseteq E$ such that $C$ is (the edge set of) a tree.
For a component $C$, we denote by $T_C \subseteq T$ the set of terminals connected by the component $C$, i.e., the set of terminals that are an endpoint of at least one of the edges of the tree $C$.
A component is called a \emph{$k$-component} if it connects at most $k$ terminals, i.e., $|T_C|\le k$.

A \emph{$k$-restricted Steiner tree} $F=\bigcupp_{C\in\Cscr} C$ is the disjoint union of a collection $\Cscr$ of $k$-components such that the hypergraph with vertex set $T$ and hyperedge set $\{ T_C : C\in \Cscr\}$ is connected. 

We denote by $\OPT$ an optimal Steiner tree solution and by $\OPT_k$ an optimal $k$-restricted Steiner tree, i.e., a $k$-restricted Steiner tree  $F$ minimizing $w(F)$.
Borchers and Du~\cite{borchers_1997_thek} showed that for large $k$ the weight $w(\OPT_k)$ of a cheapest $k$-restricted Steiner tree cannot be much larger than the weight $w(\OPT)$ of an optimal Steiner tree solution.

\begin{theorem}[\cite{borchers_1997_thek}]\label{thm:steiner_decomposition}
For any instance $(G=(V,E),T,w)$ of the Steiner tree problem and $k\in \mathbb{Z}_{\geq 2}$, we have
\[
\frac{w(\OPT_k)}{w(\OPT)} \ \le\ 1 + \frac{1}{\lfloor \log_2(k) \rfloor}\enspace.
\]
\end{theorem}

In the next section we show that for any constants $\tilde \epsilon >0$ and $k\in \mathbb{Z}_{\ge 2}$, there is a polynomial-time local search procedure that computes a Steiner tree solution $F$ with $w(F) \le (\ln 4 +\tilde\epsilon) \cdot w(\OPT_k)$.
Hence, together with Theorem~\ref{thm:steiner_decomposition}, this implies that for any $\epsilon \in (0,1]$, we can get a Steiner tree solution $F$ with $w(F) \leq (\ln 4 + \epsilon)\cdot w(\OPT)$. Indeed, this can be obtained by choosing $\tilde\epsilon = \sfrac{\epsilon}{3}$ and $k=2^{\lceil\sfrac{2 \ln(4)}{\epsilon} \rceil}$.

\subsection{The Local Search Algorithm}\label{sec:steiner_algo}

A \emph{terminal spanning tree} is a set $S \subseteq \left(\!\begin{smallmatrix} T\\ 2 \end{smallmatrix}\!\right)$ such that $(T,S)$ is a (spanning) tree.
In our local search algorithm for the Steiner tree problem we maintain
\begin{itemize}
\item a Steiner tree solution $F$,
\item a non-empty witness set $W_{f} \subseteq \left(\!\begin{smallmatrix} T\\ 2 \end{smallmatrix}\!\right)$ for all $f \in F$ such that 
\begin{itemize}
\item the union $S=\bigcup_{f\in F} W_{f}$ of the witness sets is a terminal spanning tree, and 
\item for every edge $e=\{v,w\}\in S$, the set $\{ f\in F: e\in W_f\}$ contains a $v$-$w$ path.
\end{itemize}
\end{itemize}
Because $S$ is a terminal spanning tree throughout the algorithm, the latter property of the witness sets guarantees that $F$ indeed remains a feasible Steiner tree solution throughout the algorithm.
Moreover, if the witness set $W_f$ for an edge $f\in F$ is empty, we can remove $f$ from $F$ while maintaining a feasible Steiner tree solution.
This concept of witness sets has been introduced in ~\cite{byrka_2013_steiner}.
\bigskip

Similar to \cite{goemans_2012_matroids} and our WTAP algorithm from Section~\ref{sec:WTAP}, we define a weight function $\overline{w}$ where we distribute the weight $w(f)$ equally on the edges in the witness set $W_f$.
Formally, for an edge $e$ contained in the terminal spanning tree $S = \bigcup_{f\in F} W_f$, we define 
\[
\overline{w}(e) \coloneqq \sum_{f\in F: e\in W_f} \frac{1}{|W_f|} \cdot w(f)\enspace .
\]
Then we have $\overline{w}(S)=w(F)$ because
\[
\overline{w}(S)\ =\ \sum_{e\in S} \sum_{f\in F: e\in W_f} \frac{1}{|W_f|} \cdot w(f) \ =\ \sum_{f\in F}\sum_{e\in W_f}\frac{1}{|W_f|} \cdot w(f) \ =\  \sum_{f\in F} w(f) \ =\ w(F)\enspace.
\]
Moreover, we define the potential function $\Phi$ to be 
\[
\Phi(F) \coloneqq \sum_{f\in F} H_{|W_f|} \cdot w(f) \enspace,
\]
where we again use the notation $H_q\coloneqq \sum_{i=1}^q \frac{1}{i}$.
Our algorithm will iteratively make local improvement steps that decrease the potential $\Phi(F)$.
Essentially the same potential function has been used in~\cite{goemans_2012_matroids} in the analysis of a different $(\ln 4+\epsilon)$-approximation algorithm for Steiner tree.
\bigskip

Let us now discuss how we choose the witness sets $W_f$ for $f\in F$.
We will do this in the same way as \textcite{byrka_2013_steiner}.
The below lemma captures the key properties of the witness sets that we will need to analyze our local search procedure.

\begin{lemma}[\cite{byrka_2013_steiner, goemans_2012_matroids}]\label{lem:witness_sets}
For any component $C\subseteq E$, we can efficiently find a tree $S_C \subseteq \left(\!\begin{smallmatrix} T_C\\ 2 \end{smallmatrix}\!\right)$ spanning $T_C$ and sets $W_f \subseteq S_C$ for all $f\in C$ such that 
\begin{itemize}\itemsep2pt
\item $\Phi(C) = \sum_{f\in C} H_{|W_f|} \cdot w(f) \le \ln(4)\cdot w(C)$, and 
\item for every edge $e=\{v,w\}\in S_C$, the set $\{ f\in C: e\in W_f\}$ contains a $v$-$w$ path.
\end{itemize}
\end{lemma}
\begin{proof}
Once we fixed a tree $S_C \subseteq \left(\!\begin{smallmatrix} T_C\\ 2 \end{smallmatrix}\!\right)$ spanning $T_C$, we set 
\[
W_f\coloneqq\{\{a,b\}\in S_C : f\text{ is contained in the unique $a$-$b$ path in $C$}\}\enspace.
\]
It follows from \cite{byrka_2013_steiner} that there exists a choice of $S_C$ such that $\Phi(C) \le \ln(4) \cdot w(C)$.
More precisely, it was shown in~\cite{byrka_2013_steiner} that one can choose a random tree $S_C \subseteq \left(\!\begin{smallmatrix} T_C\\ 2 \end{smallmatrix}\!\right)$ spanning $T_C$ such that for every edge $f \in C$, we have $\mathbb{P}[|W_{f}|\le q] \ge \sum_{i=1}^q \frac{1}{2^i}$ for all $q\in \mathbb{Z}_{\ge 0}$.%
\footnote{
\textcite{byrka_2013_steiner} call the tree $S_C$ a witness tree and denote it by $W$. 
We apply their construction of the witness tree (in Section~5 of \cite{byrka_2013_steiner}) to the component $C$, i.e., to the Steiner tree connecting the terminals in $T_C$.
For the bound on the cardinality of the witness set $W_{\ell}$, see Lemma~18 in \cite{byrka_2013_steiner}.
}
Then $\mathbb{E}[|W_{f}|] \le \sum_{i=1}^{\infty} \frac{1}{2^i} \cdot H_i = \ln 4$ for every edge $f\in C$, implying $\mathbb{E}[\Phi(C)]\le \ln(4)\cdot w(C)$.

Finally, it was shown in~\cite{goemans_2012_matroids} how to efficiently find, among a large class of trees $S_C \subseteq \left(\!\begin{smallmatrix} T_C\\ 2 \end{smallmatrix}\!\right)$ spanning $T_C$, which includes the ones considered by~\cite{byrka_2013_steiner}, a tree $S_C$ that minimizes $\Phi(C)$ through a dynamic program (Lemma~B.3 in~\cite{goemans_2012_matroids}).
Hence, such a tree $S_C$ leads to a potential $\Phi(C)$ that satisfies $\Phi(C) \leq \ln(4)\cdot w(C)$ as desired.
\end{proof}
We remark that if the starting solution of our local search procedure is a $k$-restricted Steiner tree, we need to apply Lemma~\ref{lem:witness_sets} only to $k$-components with constant $k$.
Then one does not need to use the dynamic program from~\cite{goemans_2012_matroids} to compute $S_C$, but one can simply enumerate over all possible trees $S_C \subseteq \left(\!\begin{smallmatrix} T_C\\ 2 \end{smallmatrix}\!\right)$ spanning $T_C$ to find the tree $S_C$ that leads to a minimum  potential $\Phi(C)=\sum_{f\in C} H_{|W_f|} \cdot w(f)$.
\bigskip

In the following we denote by $S_C \subseteq \left(\!\begin{smallmatrix} T_C\\ 2 \end{smallmatrix}\!\right)$ the tree that we obtain from Lemma~\ref{lem:witness_sets}.
For a terminal spanning tree $S$ and a component $C\subseteq E$, we define $\Drop^{\overline{w}}_S(C)\subseteq S$ to be a set in
\[
\mathrm{argmax} \{ \overline{w}(D) : D\subseteq S \text{ such that }(S \setminus D) \cup S_C \text{ is a terminal spanning tree} \}\enspace,
\]
i.e., $\Drop^{\overline{w}}_S(C)$ is a maximum weight set with respect to $\overline{w}$ that we can remove from $S$ when adding $S_C$.
We observe that, for any $D\subseteq S$, the set $(S \setminus D) \cup S_C$ is a terminal spanning tree if and only if the graph $(T, S\setminus D)/T_C$ is connected.
Thus, $\Drop^{\overline{w}}_S(C)$ depends only on the set $T_C$ of terminals connected by the component $C$.
\bigskip

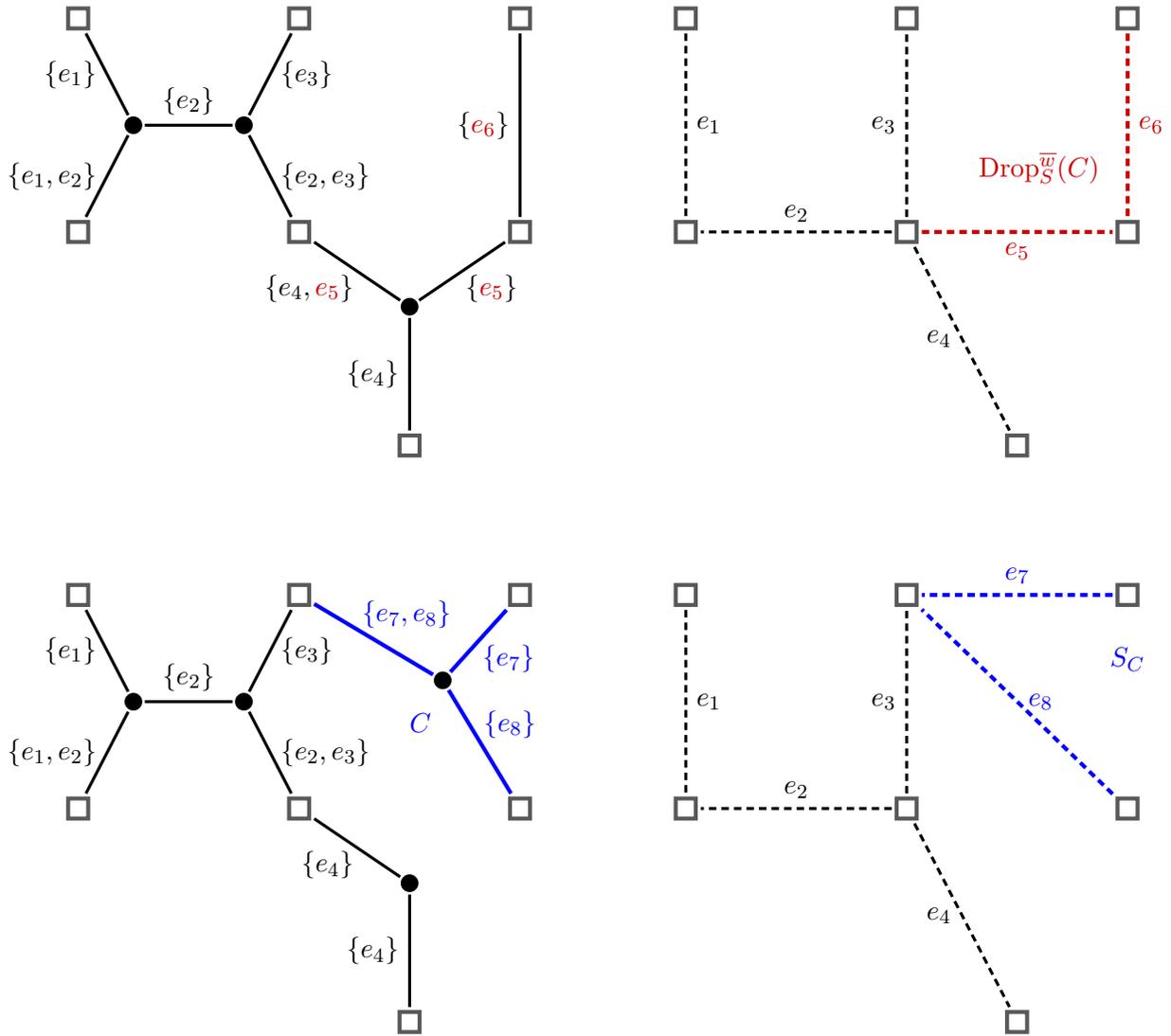
\begin{figure}
\begin{center}
\begin{tikzpicture}[xscale=1.55,yscale=3]

\tikzset{
  prefix node name/.style={%
    /tikz/name/.append style={%
      /tikz/alias={#1##1}%
    }%
  }
}

\tikzset{steiner/.style={
fill=black,circle,inner sep=0em,minimum size=0pt, inner sep=2.5pt, outer sep=1pt}
}

\newcommand\term[2][]{

\begin{scope}[prefix node name=#1]

\begin{scope}[every node/.style={ultra thick,draw=gray!70!black,fill=none,rectangle,minimum size=0pt, inner sep=4pt, outer sep=2pt}]

\node (0) at (0,1) {};
\node (1) at (0,2) {};
\node (2) at (2,2) {};
\node (3) at (2,1) {};
\node (4) at (3,0) {};
\node (5) at (4,1) {};
\node (6) at (4,2) {};

\end{scope}

\end{scope}

}%

\begin{scope}[shift={(0,0)}]
\term[f1-]{}
\begin{scope}[every node/.style={steiner}]
\node (f1-a) at (0.5,1.5) {};
\node (f1-b) at (1.5,1.5) {};
\node (f1-c) at (3, 0.65) {};
\end{scope}
\end{scope}

\begin{scope}[shift={(5.5,0)}]
\term[s1-]{}

\node[darkred] () at (3.2,1.3) {$\Drop^{\overline{w}}_S(C)$};
\end{scope}

\begin{scope}[shift={(0,-2.7)}]
\term[f2-]{}

\node[blue] () at (3.1,1.4) {$C$};

\begin{scope}[every node/.style={steiner}]
\node (f2-a) at (0.5,1.5) {};
\node (f2-b) at (1.5,1.5) {};
\node (f2-c) at (3, 0.65) {};
\node (f2-d) at (3.3, 1.6) {};
\end{scope}

\end{scope}

\begin{scope}[shift={(5.5,-2.7)}]
\term[s2-]{}

\node[blue] () at (4,1.7) {$S_C$};
\end{scope}

\begin{scope}[very thick, black]
\draw (f1-0) -- node[left] () {$\{e_1 , e_2 \}$} (f1-a);
\draw (f1-1) -- node[left] () {$\{e_1 \}$} (f1-a);
\draw (f1-a) -- node[above] () {$\{e_2 \}$} (f1-b);
\draw (f1-2) -- node[right] () {$\{e_3 \}$} (f1-b);
\draw (f1-3) -- node[right] () {$\{e_2 , e_3 \}$} (f1-b);

\draw (f1-3) -- node[below left=-4pt] () {$\{e_4 , \textcolor{darkred}{e_5} \}$} (f1-c);
\draw (f1-4) -- node[left] () {$\{e_4 \}$} (f1-c);
\draw (f1-5) -- node[below right=-4pt] () {$\{\textcolor{darkred}{e_5}\}$} (f1-c);

\draw (f1-5) -- node[left] () {$\{\textcolor{darkred}{e_6} \}$} (f1-6);

\draw (f2-0) -- node[left] () {$\{e_1 , e_2 \}$} (f2-a);
\draw (f2-1) -- node[left] () {$\{e_1 \}$} (f2-a);
\draw (f2-a) -- node[above] () {$\{e_2 \}$} (f2-b);
\draw (f2-2) -- node[right] () {$\{e_3 \}$} (f2-b);
\draw (f2-3) -- node[right] () {$\{e_2 , e_3 \}$} (f2-b);

\draw (f2-3) -- node[below left=-4pt] () {$\{e_4 \}$} (f2-c);
\draw (f2-4) -- node[left] () {$\{e_4 \}$} (f2-c);

\begin{scope}[blue, ultra thick] %
\draw (f2-2) -- node[above right=-0pt and -9pt] () {$\{e_7, e_8 \}$} (f2-d);
\draw (f2-5) -- node[above right=-4pt] () {$\{e_8 \}$} (f2-d);
\draw (f2-6) -- node[below right=-4pt] () {$\{e_7 \}$} (f2-d);

\end{scope}
\end{scope}

\begin{scope}[very thick, densely dashed]

\draw (s1-1) -- node[right] () {$e_1$} (s1-0);
\draw (s1-3) -- node[above] () {$e_2$} (s1-0);
\draw (s1-3) -- node[left] () {$e_3$} (s1-2);
\draw (s1-3) -- node[left] () {$e_4$} (s1-4);
\draw[darkred, ultra thick] (s1-3) -- node[below] () {$e_5$}(s1-5);
\draw[darkred, ultra thick] (s1-6) -- node[right] () {$e_6$} (s1-5);

\draw (s2-1) -- node[right] () {$e_1$} (s2-0);
\draw (s2-3) -- node[above] () {$e_2$} (s2-0);
\draw (s2-3) -- node[left] () {$e_3$} (s2-2);
\draw (s2-3) -- node[left] () {$e_4$} (s2-4);

\begin{scope}[blue, ultra thick] %
\draw (s2-5) -- node[right] () {$e_8$} (s2-2);
\draw (s2-6) -- node[above] () {$e_7$}(s2-2);
\end{scope}

\end{scope}

\end{tikzpicture}
 \end{center}
\caption{Example of a local improvement step.
Terminals are shown as squares, non-terminals as circles.
The left column shows the current Steiner tree solution $F$ and the right column the terminal spanning tree $S$ before (top) and after (bottom) the improvement step.
The newly added component $C$ and the tree $S_C$ are shown in blue, the set $\Drop^{\overline{w}}_S(C)$ in red.
The witness sets $W_f$ for $f\in F$ are written next to the edges in $F$.
\label{fig:example_steiner}}
\end{figure}

In a local improvement step of our algorithm, we will select a $k$-component $C$, add $S_C$ to $S$, and remove $\Drop^{\overline{w}}_S(C)$ from $S$. 
Then we will remove all $f\in F$ from $F$ for which the witness set $W_f$ became empty.
See Figure~\ref{fig:example_steiner} for an example.
In such an improvement step, the potential $\Phi(F)$ will decrease by at least $\overline{w}(\Drop^{\overline{w}}_S(C)) - \ln(4)\cdot w(C)$ as we show below (Lemma~\ref{lem:potential_decrease_steiner}), and we therefore select a $k$-component $C$ maximizing $\overline{w}(\Drop^{\overline{w}}_S(C)) - \ln(4)\cdot w(C)$.
A formal description of our local search algorithm is given in Algorithm~\ref{algo:steiner_local_search}.

In the following we fix constants $0 < \epsilon \le 1$ and $k\in \mathbb{Z}_{\ge 2}$, and we define $n\coloneqq |V|$.

\begin{algorithm2e}[H]
\KwIn{A Steiner tree instance $(G=(V,E),T,w)$.}
\KwOut{A Steiner tree solution $F$ for the terminal set $T$ with $w(F)\le (\ln4+\epsilon)\cdot w(\OPT_k)$.}
\vspace*{2mm}
\begin{enumerate}[label=\arabic*.,ref=\arabic*,rightmargin=7mm]\itemsep4pt
\item\label{item:initialize_steiner}
Let $F\subseteq E$ be an arbitrary Steiner tree. \\
Define witness sets $W_f$ for all $f\in F$ by applying Lemma~\ref{lem:witness_sets} to $F$.
\item\label{item:local_step_steiner} Iterate the following as long as $\Phi(F)$ decreases by at least a factor $ \left( 1- \frac{\epsilon}{2H_n\cdot\ln(4)\cdot|T|}\right)$.
\begin{itemize}\itemsep1pt
\item \textbf{Select a best component:} Choose a $k$-component $C\subseteq E$ that maximizes $\overline{w}(\Drop^{\overline{w}}_S(C)) - \ln(4) \cdot w(C)$, where $S=\bigcup_{f\in F} W_f$. (See Lemma~\ref{lem:find_k_component}.)
\item \textbf{Remove $\mathrm{\mathbf{Drop}}$:} Replace the witness set $W_{f}$ by $W_{f} \setminus \Drop^{\overline{w}}_S(C)$ for all $f \in F$. \\
      If for some edge $f\in F$, the witness set $W_{f}$ becomes empty, remove $f$ from $F$.
\item \textbf{Add the new component:} Add $C$ to $F$ and define witness sets $W_f$ for all $f\in C$ by applying Lemma~\ref{lem:witness_sets}.
\end{itemize}
\item Return $F$.
\end{enumerate}
\caption{Local search algorithm for Steiner tree}\label{algo:steiner_local_search}
\end{algorithm2e}
Note that in step~\ref{item:initialize_steiner} of Algorithm~\ref{algo:steiner_local_search}, we apply Lemma~\ref{lem:witness_sets} to the whole Steiner tree, which we can view as a single component.
Here we assume without loss of generality that $F$ is (the edge set of) a tree; otherwise we can remove some edges from $F$ while maintaining a Steiner tree solution.
Alternatively, if the starting solution $F$ computed in step~\ref{item:initialize_steiner} is a $k$-restricted Steiner tree, we can simply apply Lemma~\ref{lem:witness_sets} to every $k$-component $C\in \Cscr$ of the $k$-restricted Steiner tree $F=\bigcupp_{C\in\Cscr} C$ separately.\footnote
{ 
If $F$ is not inclusionwise minimal, it might happen that $S=\bigcup_{f\in F} W_f$ contains a terminal spanning tree, but is not a terminal spanning tree itself. 
In this case, we can remove edges from $S$ to turn it into a terminal spanning tree.
Alternatively, we can first remove edges from $F$ to turn it into an inclusionwise minimal $k$-restricted Steiner tree before applying Lemma~\ref{lem:witness_sets} to every $k$-component.
}
As mentioned after Lemma~\ref{lem:witness_sets}, this avoids applying Lemma~\ref{lem:witness_sets} to components connecting more than a constant number of terminals.
\bigskip 

Let us now turn to the analysis of Algorithm~\ref{algo:steiner_local_search}.
First, we show that, in every iteration of the algorithm, the potential $\Phi(F)$ indeed decreases by at least $\overline{w}(\Drop^{\overline{w}}_S(C)) - \ln(4)\cdot w(C)$.

\begin{lemma}\label{lem:potential_decrease_steiner}
Whenever we select a $k$-component $C$ in some iteration  of step~\ref{item:local_step_steiner} of Algorithm~\ref{algo:steiner_local_search}, then $\Phi(F)$ decreases by at least $\overline{w}(\Drop^{\overline{w}}_S(C)) - \ln(4) \cdot w(C)$ in this iteration.
\end{lemma}
\begin{proof}
To simplify the notation, we write $\Drop \coloneqq \Drop^{\overline{w}}_S(C)$.
We first bound the decrease of the potential by removing $\Drop$ from all witness sets $W_f$ with $f\in F$.
We have
\begin{align*}
\sum_{f\in F} H_{|W_f|}\cdot w(f) - \sum_{f\in F} H_{|W_f\setminus \Drop|}\cdot w(f)
=&\ \sum_{f\in F}\ \sum_{i=|W_f\setminus \Drop|+1}^{|W_f|} \frac{1}{i}\cdot w(f) \\
\ge& \sum_{f\in F}\ \sum_{e\in W_f\cap\Drop} \frac{1}{|W_f|}\cdot w(f) \\
=&\ \sum_{e\in \Drop}\ \sum_{f\in F:e\in W_f} \frac{1}{|W_f|}\cdot w(f) \\
=&\ \overline{w}(\Drop)\enspace.
\end{align*}
Moreover, by Lemma~\ref{lem:witness_sets}, we have 
$\Phi(C)=\sum_{f\in C} H_{|W_f|} \cdot w(f) \le \ln(4)\cdot w(C)$. Thus, adding $C$ to $F$ increases the potential $\Phi(F)$ by at most $\ln(4)\cdot w(C)$ .
\end{proof}

Next, we show that we can find a best component in step~\ref{item:local_step_steiner} of our algorithm efficiently.

\begin{lemma}\label{lem:find_k_component}
Let $k\in \mathbb{Z}_{\ge 2}$ be a constant.
Given a Steiner tree instance $(G=(V,E),T,w)$, a terminal spanning tree $S$, and weights $\overline{w}: S\to \mathbb{R}_{\ge 0}$, we can efficiently compute a $k$-component $C\subseteq E$ 
maximizing $\overline{w}(\Drop^{\overline{w}}_S(C)) - \ln(4) \cdot w(C) $.
\end{lemma}
\begin{proof}
Because $\Drop^{\overline{w}}_S(C)$ depends only on the terminal set $T_C$ connected by the component $C$, the following yields an optimal $k$-component.
We enumerate over all subsets $T_C \subseteq T$ with $|T_C|\le k$ and compute a cheapest Steiner tree with terminal set $T_C$, which is possible in polynomial time because $k$ is constant; see, e.g.,~\cite{dreyfus_1971_steiner}.
Among these components, we return the one maximizing $\overline{w}(\Drop^{\overline{w}}_S(C)) - \ln(4) \cdot w(C) $.
\end{proof}

\begin{lemma}\label{lem:steiner_feasible}
When Algorithm~\ref{algo:steiner_local_search} terminates, the edge set $F$ is a Steiner tree solution.
\end{lemma}
\begin{proof}
At the end of step~\ref{item:initialize_steiner}, $F$ is a Steiner tree solution and hence, by Lemma~\ref{lem:witness_sets}, the set $S=\cup_{f\in F}W_f$ is a terminal spanning tree.
By the definition of $\Drop^{\overline{w}}_S(C)$, the edge set $S$ remains a terminal spanning tree throughout the algorithm.
When we add an edge $e=\{a,b\}$ to $S$, then $\{f\in F: e\in W_f \}$ contains an $a$-$b$ path $P$.
Because an edge $f$ is only removed from $F$ when its witness set $W_f$ becomes empty, the path $P$ remains in $F$ until $e$ is removed from $S$.
Therefore, because $S$ connects all terminals throughout the algorithm, also $F$ connects all terminals throughout the algorithm.
\end{proof}

Let us now analyze the approximation ratio of our algorithm.
We first prove a lower bound on the decrease of $\Phi(F)$ in a single iteration.
For this we use a well-known block exchange property of matroids, stated in the lemma below. 
We will apply this result to the matroid whose bases are the terminal spanning trees.
\begin{lemma}[\cite{greene_1975_some}]\label{lem:matroids_block_exchange}
Let $\mathcal{M}$ be a matroid and let $B_1,B_2$ be bases of $\mathcal{M}$. 
Let $\mathcal{P}_1$ be a partition of $B_1$.
Then there exists a partition $\mathcal{P}_2$ of $B_2$ and a bijection $\phi: \mathcal{P}_1 \to \mathcal{P}_2$ such that for each $X \in \mathcal{P}_1$, the set $(B_2 \setminus \phi(X))\cup X$ is a basis of the matroid $\Mscr$.
\end{lemma}

Lemma~\ref{lem:matroids_block_exchange} was proven in~\cite{greene_1975_some}; see also (42.15) in \cite{schrijver_2003_combinatorial}.\footnote
{
Theorem 3.3 in \cite{greene_1975_some} is a slightly different but equivalent version of Lemma~\ref{lem:matroids_block_exchange}, requiring that $(B_1\setminus X)\cup \phi(X)$ is a basis of $\mathcal{M}$ instead of $(B_2\setminus \phi(X))\cup X$, for every $X\in \mathcal{P}_1$. Lemma~\ref{lem:matroids_block_exchange} immediately follows from the one in~\cite{greene_1975_some} by applying the version in~\cite{greene_1975_some} to the dual matroid of the matroid $\mathcal{M}\vert_{B_1\cup B_2}$, where $\mathcal{M}\vert_{B_1\cup B_2}$ is the restriction of $\mathcal{M}$ to the elements $B_1\cup B_2$.
}
Lemma~\ref{lem:good_improvement_exists_steiner} below (together with Lemma~\ref{lem:potential_decrease_steiner}) provides a lower bound on the decrease of $\Phi(F)$ in a single iteration.
In particular, it immediately implies that, as long as our current solution $F$ has a weight strictly larger than $\ln(4)\cdot w(\OPT_k)$, the potential $\Phi(F)$ decreases, i.e., our algorithm makes progress.

\begin{lemma}\label{lem:good_improvement_exists_steiner}
In every iteration of Algorithm~\ref{algo:steiner_local_search}, we have 
\begin{equation*}
\overline{w}(\Drop^{\overline{w}}_S(C)) - \ln(4) \cdot w(C) \ \ge\ \frac{1}{|T|} \cdot \Big(w(F) - \ln(4) \cdot w(\OPT_k)\Big)\enspace.
\end{equation*}
\end{lemma}
\begin{proof}
Because $\OPT_k$ is a $k$-restricted Steiner tree, we can write
$\OPT_k=\bigcupp_{C\in\Cscr} C$ as the disjoint union of a collection $\Cscr$ of $k$-components such that the hypergraph with vertex set $T$ and edge set $\{ T_C : C\in \Cscr\}$ is connected. 
Then the disjoint union of the trees $S_C$ over all components in $\Cscr$ contains a terminal spanning tree.
In fact, it even is a terminal spanning tree because $\OPT_k$ is an optimal, and thus minimal, $k$-resticted Steiner tree.

Hence, by Lemma~\ref{lem:matroids_block_exchange} applied to the matoid whose bases are the terminal spanning trees, there is a partition $\Pscr$ of $S$ and a bijection $\Phi: \Cscr \to \Pscr$ such that, for each $C\in \Cscr$, the set $(S\setminus \phi(C)) \cup S_C$ is a terminal spanning tree.
Therefore, for each $C\in\Cscr$ we have $\overline{w}(\Drop^{\overline{w}}_S(C)) \ge \overline{w}(\phi(C))$ by the definition of $\Drop^{\overline{w}}_S(C)$.
Because $\sum_{C\in \Cscr}\overline{w}(\phi(C)) = \overline{w}(S)\ =\ w(F)$ and $\sum_{C\in \Cscr} w(C) = w(\OPT_k)$, this implies
\begin{align*}
\max_{C\in\Cscr} \Big(\overline{w}(\Drop^{\overline{w}}_S(C))-\ln(4) \cdot w(C)\Big)\ \ge&\ \max_{C\in\Cscr} \Big(\overline{w}(\phi(C))-\ln(4) \cdot w(C)\Big)\\
\ge&\ \frac{1}{|\Cscr|}\sum_{C\in \Cscr} \Big(\overline{w}(\phi(C))-\ln(4) \cdot w(C)\Big)\\
 \ge&\ \frac{1}{|T|}\cdot\Big(w(F) - \ln(4)\cdot w(\OPT_k)\Big)\enspace,
\end{align*}
where we used $|\Cscr| \le |T|$ for the last inequality.
\end{proof}

Using the above lower bound on the progress we make in a single iteration of our local search procedure, we can now prove that Algorithm~\ref{algo:steiner_local_search} indeed has the claimed approximation guarantee.

\begin{lemma}\label{lem:apx_guarantee_steiner}
When Algorithm~\ref{algo:steiner_local_search} terminates, it returns a Steiner tree solution $F$ with $w(F)\le (\ln(4)+\epsilon)\cdot w(\OPT_k)$.
\end{lemma}
\begin{proof}
By Lemma~\ref{lem:potential_decrease_steiner} and Lemma~\ref{lem:good_improvement_exists_steiner}, we have that, in every iteration, the potential $\Phi(F)$ decreases by at least $\frac{1}{|T|} \cdot \left(w(F) - \ln(4) \cdot w(\OPT_k)\right)$.
Thus, when the algorithm terminates, it returns a Steiner tree solution $F\subseteq E$ that satisfies
\[ 
\frac{\epsilon}{2H_n\cdot\ln(4)\cdot |T|} \cdot \Phi(F) > \frac{1}{|T|} \cdot \left(w(F) - \ln(4) \cdot w(\OPT_k)\right) \enspace.
\]
(Note that Lemma~\ref{lem:steiner_feasible} guarantees that $F$ is a Steiner tree solution.)
Using $\Phi(F)\le H_n \cdot w(F)$, this implies
\[
\left(1-\frac{\epsilon}{2\cdot\ln(4)}\right) \cdot w(F) < \ln(4) \cdot w(\OPT_k)
\]
and thus  
\begin{align*}
w(F) &<    \left(1-\frac{\epsilon}{2\cdot \ln 4}\right)^{-1}\cdot \ln(4) \cdot w(\OPT_k)  \\
     &=    \left(1 + \frac{\epsilon}{2\ln(4)-\epsilon}\right)\cdot \ln(4) \cdot w(\OPT_k) \\
     &\leq \left(1 + \frac{\epsilon}{\ln(4)}\right)\cdot \ln(4) \cdot w(\OPT_k)           \\
     &=    \left(\ln(4) + \epsilon \right)\cdot w(\OPT_k)\enspace,
\end{align*}
where the last inequality uses $\epsilon\leq \ln(4)$.
\end{proof}

Finally, we show that our local search procedure terminates in polynomial time, which follows from an analysis analogous to the one we applied in the context of WTAP to derive Lemma~\ref{lem:bound_iter_WTAP}.
Note that the initial Steiner tree solution $F_0$ computed in step~\ref{item:initialize} of Algorithm~\ref{algo:steiner_local_search} has weight $w(F_0)\leq w(E)$; thus, the bound stated in the lemma below is indeed polynomial independently of the starting solution $F_0$.
\begin{lemma}
Algorithm~\ref{algo:steiner_local_search} terminates after at most $\ln\left(\frac{H_n \cdot w(F_0)}{w(\OPT)}\right)\cdot \frac{2H_n \cdot \ln(4)\cdot |V|}{\epsilon}$ iterations, where $F_0\subseteq E$ is the initial Steiner tree computed in step~\ref{item:initialize_steiner} of Algorithm~\ref{algo:steiner_local_search}.
\end{lemma}
\begin{proof}
At the beginning of Algorithm~\ref{algo:steiner_local_search} we have $\Phi(F)=\Phi(F_0)\le H_n \cdot w(F_0)$.
Because the potential $\Phi(F)$ decreases by a factor of at least $\left(1- \frac{\epsilon}{2H_n\cdot\ln(4)\cdot |T|}\right)$ in every iteration and because $\Phi(F)\ge w(F)\ge w(\OPT)$ throughout the algorithm, the number of iterations is at most
\begin{align*}
\log_{(1-\sfrac{\epsilon}{(2H_n\cdot\ln(4)\cdot |T|)})^{-1}} \left(\frac{H_n \cdot w(F_0)}{w(\OPT)}\right) \ =&\ \ln\left(\frac{H_n\cdot w(F_0)}{w(\OPT)}\right)\cdot \frac{1}{-\ln(1-\sfrac{\epsilon}{(2H_n\cdot\ln(4)\cdot |T|)})}\\[2mm]
 \le&\ \ln\left(\frac{H_n\cdot w(F_0)}{w(\OPT)}\right)\cdot \frac{2H_n\cdot\ln(4)\cdot |T|}{\epsilon}\enspace,
\end{align*}
where we used $\ln(1+x) \le x$ for $x> -1$.
\end{proof}

Finally, we note that, analogous to the algorithms in~\cite{byrka_2013_steiner,goemans_2012_matroids}, Algorithm~\ref{algo:steiner_local_search} can be improved for Steiner Tree problems restricted to particular graph topologies.
For example when the underlying graph is quasi-bipartite, i.e., non-terminal nodes are pairwise non-adjacent.
In such cases, one can get a lower potential for a component $C$ than $\ln(4)\cdot w(C)$, i.e., one can strengthen the $\ln 4$ factor in Lemma~\ref{lem:witness_sets}.
The only modification necessary in Algorithm~\ref{algo:steiner_local_search} to obtain improved factors in such cases, is to select a $k$-component that maximizes $\overline{w}(\Drop^{\overline{w}}_S(C)) - \Phi(C)$ instead of $\overline{w}(\Drop^{\overline{w}}_S(C)) - \ln 4 \cdot w(C)$.
This will lead to the same improved factors as with the procedures in~\cite{byrka_2013_steiner,goemans_2012_matroids}.

\subsection{An LP based Analysis}\label{sec:steiner_lp}

The $(\ln 4+\epsilon)$-approximation algorithm by \textcite{byrka_2013_steiner} is based on a linear programming relaxation, called the directed component LP.
The variables in this LP relaxation correspond to directed components, i.e., to pairs $(C,t)$ where $C$ is a component and $t\in T_C$ is a terminal that we interpret as the root of the component $C$.
We write 
\[
\vec{\mathcal{C}} \coloneqq \left\{ (C,t) \colon C\subseteq E\text{ is a component with }t\in T_C \right\}
\]
 to denote the set of all directed components.
Then, for a terminal set $R\subseteq T$, we denote by 
\[
 \delta^-_{\vec{\Cscr}}(R) \coloneqq \left\{ (C,t)\in \vec{\Cscr} \colon t\notin R,\ T_C\cap R \ne \emptyset \right\}
\]
the set of directed components that enter $R$.
The directed component relaxation can now be stated as follows, where $r\in T$ is an arbitrary fixed terminal:
\begin{equation}\label{eq:dcr}
\begin{array}{rr@{\;}c@{\;}ll}
\min & \multicolumn{3}{c}{\displaystyle\sum_{(C,t)\in\vec{\Cscr}} w(C) \cdot x_{C,t} }\\
&\displaystyle\sum_{(C,t)\in  \delta^-_{\vec{\Cscr}}(R)} x_{C,t} &\geq  &1 &\forall\; \emptyset \ne R\subseteq T\setminus\{r\}\\
&x_{C,t} &\ge &0 &\forall\;(C,t)\in\vec{\Cscr}\enspace.
\end{array}
\end{equation}
It is NP-hard to solve the LP~\eqref{eq:dcr} exactly as observed in~\cite{goemans_2012_matroids}, but it follows from Theorem~\ref{thm:steiner_decomposition} that for large enough $k$, the weight of an optimum solution $\mathrm{LP}_k$ of the $k$-restricted directed component LP, i.e., LP~\eqref{eq:dcr} restricted to the set 
\[
\vec{\Cscr}_k \coloneqq \left\{ (C,t) \colon C\subseteq E\text{ is a $k$-component with }t\in T_C \right\}
\]
of directed $k$-components, is at most $(1+\epsilon)\cdot w(\mathrm{LP})$, where $\mathrm{LP}$ denotes an optimal solution of the unrestricted LP \eqref{eq:dcr}.

\textcite{goemans_2012_matroids} analyzed a variant of the algorithm in~\cite{byrka_2013_steiner} and showed that this algorithm yields an $(\ln 4+\epsilon)$-approximation not only with respect to an optimal Steiner tree solution, but also with respect to the optimal value of the directed component LP, i.e., they showed that the computed solution $F$ fulfills $w(F)\le (\ln 4 +\epsilon)\cdot w(\mathrm{LP})$.

Next we show that also our local search algorithm can be analyzed with respect to the directed component LP, although we needed the LP neither in the algorithm itself nor in the proof of its approximation guarantee.

\begin{theorem}\label{thm:int_gap}
Algorithm~\ref{algo:steiner_local_search} returns a Steiner tree solution $F$ with $w(F)\le (\ln 4 +\epsilon)\cdot w(\mathrm{LP}_k)$.
\end{theorem}
\begin{proof}
Let $x$ be an optimal solution to the $k$-restricted directed component LP .
To show that Algorithm~\ref{algo:steiner_local_search} computes a solution $F$ with $w(F)\le (\ln 4 +\epsilon)\cdot w(\mathrm{LP}_k)$, we show that as long as $w(F)$ is strictly lager than $\ln(4)\cdot w(\mathrm{LP}_k)$, then there a component whose selection will improve the potential.
To this end, consider the state of Algorithm~\ref{algo:steiner_local_search} at the beginning of an iteration in step~\ref{item:local_step_steiner}.
By the bridge lemma (Lemma~11 in \cite{byrka_2013_steiner}), we have
\begin{equation*}
\overline{w}(S) \le \sum_{(C,t)\in\vec{\Cscr}_k} x_{C,t} \cdot \overline{w}(\Drop^{\overline{w}}_S(C))\enspace,
\end{equation*}
where $S=\cup_{f\in F} W_f$, as usual.
Because $\overline{w}(S)=w(F)$ and 
\[
\sum_{(C,t)\in\vec{\Cscr}_k} x_{C,t} \cdot w(C) = w(\mathrm{LP}_k)\enspace,
\] we obtain the following LP-based version of Lemma~\ref{lem:good_improvement_exists_steiner}:
\begin{align*}
\max_{(C,t)\in\vec{\Cscr}_k} \left(\overline{w}(\Drop^{\overline{w}}_S(C))-\ln(4) \cdot w(C)\right)\ \ge&\
 \frac{1}{x\big(\vec{\Cscr}_k \big)}\sum_{(C,t)\in\vec{\Cscr}_k} x_{C,t} \cdot \Big(\overline{w}(\Drop^{\overline{w}}_S(C))-\ln(4) \cdot w(C)\Big)\\
 \ge&\ \frac{1}{|T|}\cdot\Big(w(F) - \ln(4)\cdot w(\mathrm{LP}_k)\Big)\enspace,
\end{align*}
where we used $x\big(\vec{\Cscr}_k \big)\coloneqq \sum_{(C,t)\in\vec{\Cscr}_k}\ x_{C,t} \le |T|$ for the last inequality. 
Proceeding as in the proof of Lemma~\ref{lem:apx_guarantee_steiner}, but using the above lower bound on the decrease of $\Phi(F)$ in a single iteration of Algorithm~\ref{algo:steiner_local_search} instead of Lemma~\ref{lem:good_improvement_exists_steiner}, yields that the Steiner tree solution $F$ returned by Algorithm~\ref{algo:steiner_local_search} fulfills $w(F)\le (\ln 4+\epsilon)\cdot w(\mathrm{LP}_k)$.
\end{proof}

Theorem~\ref{thm:int_gap} immediately implies that the integrality gap of the hypergraphic LP relaxation~\eqref{eq:dcr} is at most $\ln 4$, which has first been shown in \cite{goemans_2012_matroids} with an arguably more involved reasoning.

\section{Iterative Randomized Rounding and Local Search}\label{sec:comparison}

In this section we discuss the relation of the iterative randomized rounding technique from \cite{byrka_2013_steiner} and our new local search algorithm.
In particular, we explain why it seems difficult to apply the iterative rounding technique to WTAP, despite the fact that for the Steiner tree problem both techniques yield the same approximation ratio.

The iterative randomized rounding algorithm from \cite{byrka_2013_steiner} solves the $k$-restricted directed component LP, samples a directed component $(C,t)$ proportional to the value of its LP variable $x_{C,t}$, contracts $C$, and iterates on the resulting residual instance until all terminals are connected.

One way to analyze this algorithm, which was proposed in~\cite{goemans_2012_matroids}, is to show that, in any iteration of the algorithm, the expected decrease of the potential function $\Phi(\OPT_k)$ is at least the expected weight $w(C)$ of the sampled component $C$.
Because the potential is always nonnegative, this implies that the expected weight of the resulting Steiner tree is at most $\Phi(\OPT_k) \le \ln(4)\cdot w(\OPT_k)$.
To prove that the expected decrease of the potential is at least the expected weight of the contracted component, one can use essentially the same argument that we used to show the existence of a good local improvement step (Lemma~\ref{lem:potential_decrease_steiner} and Lemma~\ref{lem:good_improvement_exists_steiner}).

The above discussion shows the close relation of the analysis of the iterative randomized rounding algorithm from \cite{byrka_2013_steiner, goemans_2012_matroids} and the analysis of our local search algorithm.
Moreover, we have seen that, for the Steiner tree problem, both techniques are equally strong in terms of the achieved approximation guarantee.
Nevertheless, it is highly unclear how one could design a $(1.5+\epsilon)$-approximation for WTAP using iterative randomized rounding.
One important reason for this is the following crucial difference between the Steiner tree problem and WTAP.
For the Steiner tree problem we could show that the partition $\Cscr$ of $\OPT_k$ into $k$-components fulfills
\begin{equation}\label{eq:decomposition_steiner_tree}
\sum_{C\in\Cscr} \overline{w}(\Drop^{\overline{w}}_S(C)) \ge \overline{w}(S)
\end{equation}
for any terminal spanning tree $S$ and any non-negative weight function $\overline{w} : S \to \mathbb{R}_{\ge 0}$.
For WTAP, we showed that for every WTAP solution $U$ for which the paths $P_u$ with $u\in U$ are disjoint and any weight function $\overline{w} : U \to \mathbb{R}_{\ge 0}$, there exists a partition $\Cscr$ of $\OPT$ into $k$-thin components such that
\begin{equation}\label{eq:decomposition_wtap}
\sum_{C\in\Cscr} \overline{w}(\Drop_U(C)) \ge (1-\epsilon)\cdot \overline{w}(U)\enspace.
\end{equation}
These statements for Steiner tree and WTAP played an analogous role in the analysis of our local search procedures, but the statement for WTAP is weaker in the sense that the partition $\Cscr$ of $\OPT$ into $k$-thin components crucially depends on the solution $U\subseteq L_{\mathrm{up}}$, while for Steiner tree the set $\Cscr$ of components is independent of $S$. 
For WTAP, this dependence of the decomposition on the solution $U$ is necessary (see Figure~\ref{fig:example_nonobliviousness_wtap_decomp}).

\begin{figure}[ht]
\begin{center}
\begin{tikzpicture}

\tikzset{
  prefix node name/.style={%
    /tikz/name/.append style={%
      /tikz/alias={#1##1}%
    }%
  }
}

\tikzset{
lks/.style={line width=2pt, densely dashed},
}

\def\rad{1.5cm}
\def\nb{3}
\def\hsep{10mm}
\def\ang{15}

\begin{scope}[prefix node name=f]
\begin{scope}[every node/.style={thick,draw=black,fill=black,circle,minimum size=0pt, inner sep=1.2pt, outer sep=1pt}]

\node (c) at (0,0) {};
\node (r) at (0,\rad) {};

\foreach \i in {-\nb, ..., \nb} {
  \node (\i) at (\i*\hsep,-\rad) {}; 
}

\end{scope}

\node[above] (tr) at (r) {$r$};
\node[above right] (tr) at (c) {$c$};

\node[blue] at ($(c)+(1.5,0.6)$) {$\mathrm{OPT}$};

\begin{scope}[very thick]
\draw (r) -- (c);

\foreach \i in {-\nb, ..., \nb} {
  \draw (c) -- (\i);
}
\end{scope}

\begin{scope}[lks]

\def\bnd{10}

\begin{scope}[blue]
\draw (r) to[bend right=\bnd] (c);

\foreach \i in {-\nb, ..., 0} {
\draw (\i) to[bend left=\bnd] (c);
}
\foreach \i in {1, ..., \nb} {
\draw (\i) to[bend right=\bnd] (c);
}
\end{scope}

\end{scope}[lks]
\end{scope}

\begin{scope}[prefix node name=s, xshift=9cm]

\def\sl{-1} %

\begin{scope}[every node/.style={thick,draw=black,fill=black,circle,minimum size=0pt, inner sep=1.2pt, outer sep=1pt}]

\node (c) at (0,0) {};
\node (r) at (0,\rad) {};

\foreach \i in {-\nb, ..., \nb} {
  \node (\i) at (\i*\hsep,-\rad) {}; 
}

\end{scope}

\node[above] (tr) at (r) {$r$};
\node[above right] (tr) at (c) {$c$};

\node[below] at (\sl) {$v$};

\node[red] at ($(c)+(1.5,0.6)$) {$U$};

\begin{scope}[very thick]
\draw (r) -- (c);

\foreach \i in {-\nb, ..., \nb} {
  \draw (c) -- (\i);
}
\end{scope}

\begin{scope}[lks]

\def\bnd{10}

\begin{scope}[red]
\draw (r) to[bend right=\bnd] (\sl);

\foreach \i in {-\nb, ..., 0} {
  \ifthenelse {\i=\sl
}{

}{
\draw (\i) to[bend left=\bnd] (c);
}
}
\foreach \i in {1, ..., \nb} {
\draw (\i) to[bend right=\bnd] (c);
}
\end{scope}

\end{scope}[lks]
\end{scope}

\end{tikzpicture}
 \end{center}
\caption{
A simple example showing why a decomposition with guarantees as stated in Theorem~\ref{thm:decomposition} cannot be computed without knowing the up-link solution $U$ upfront.
The graph we consider is a star graph with an arbitrary root $r$ and center vertex $c$.
On the left-hand side is a possible optimal solution, and the right-hand picture shows an up-link solution $U$.
Assume that the link $\{r,v\}\in U$ is expensive.
Hence, a good decomposition of $\OPT$ into components needs to have at least one component that covers $P_{\{r,v\}}$.
If we do not know the vertex $v$, then it is impossible to guarantee that such a component exists, except if all links of $\OPT$ are put into the same component.
However, this component would not be $O(1)$-thin if we start with a large star.
Also, using some cheap links of $\OPT$ multiple times when constructing components will not help if the link $\{r,c\}\in \OPT$ is expensive.
}\label{fig:example_nonobliviousness_wtap_decomp}
\end{figure}
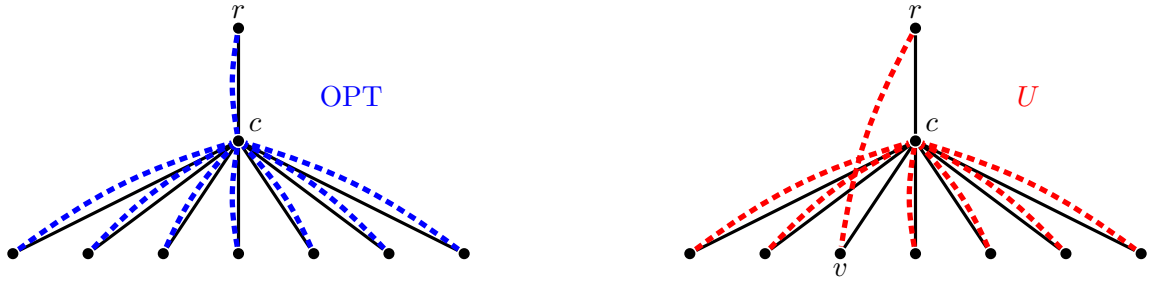

In our local search algorithm, we used \eqref{eq:decomposition_steiner_tree} and \eqref{eq:decomposition_wtap} to show that there exists a component $C\in\Cscr$ that we can use to decrease the potential $\Phi(F)$ of the current solution $F$.
In the iterative randomized rounding algorithm, we want to choose a component $C$ that we can use to decrease the potential $\Phi(\OPT)$.
Both for Steiner tree and WTAP, choosing a random component from $\Cscr$ yields a component $C$ that in expectation leads to a decrease of $\Phi(F)$.
However, in the context of WTAP, the partition $\Cscr$ depends on the unknown solution $\OPT$ and thus it seems challenging to actually find a good component $C$ to contract (or to design an LP from which we could sample $C$).
In contrast to this, in our local search procedure we know the current solution $F$ explicitly, which makes it much easier to find a component $C$ that can be used to decrease $\Phi(F)$.

\printbibliography

\end{document}